\documentclass[a4paper,onecolumn,11pt,unpublished]{quantumarticle}
\pdfoutput=1

\PassOptionsToPackage{compress}{natbib}
\usepackage[numbers]{natbib}

\usepackage{amsmath}
\usepackage{amsfonts}
\usepackage{amsthm}   %
\usepackage{mathrsfs}
\usepackage{tikz}
\usepackage[utf8]{inputenc}

\usepackage{epstopdf}
\usepackage{mathtools}
\usepackage{graphicx}
\usepackage{enumerate}
\usepackage{epsf,verbatim,amssymb,array,multicol,multirow}  
\usepackage{psfrag,bm,xspace}
\usepackage{hhline}
\usepackage{xstring}
\usepackage{ifthen}
\usepackage{booktabs}

\usepackage{xparse}
\usepackage{physics}

\usepackage{tikz}
\usetikzlibrary{quantikz}

\usepackage{algorithm}      %
\usepackage[noend]{algpseudocode} %

\ifdefined \theorem 
\else
  \newtheorem{theorem}{Theorem}
\fi

\ifdefined \lemma 
\else
  \newtheorem{lemma}{Lemma}
  
\fi

\ifdefined \corollary 
\else
\newtheorem{corollary}{Corollary}
\fi

\newtheorem{fact}{Fact}

\ifdefined \proposition 
\else
\newtheorem{proposition}{Proposition}
\fi

\ifdefined \definition 
\else
\newtheorem{definition}{Definition}
\fi

\ifdefined \example 
\else
\newtheorem{example}{Example}
\fi

\ifdefined \remark 
\else
\newtheorem{remark}{Remark}
\fi

\ifdefined \conjecture 
\else
\newtheorem{conjecture}{Conjecture}
\fi

\newtheorem{innercustomgeneric}{\customgenericname}
\providecommand{\customgenericname}{}
\newcommand{\newcustomtheorem}[2]{%
  \newenvironment{#1}[1]
  {%
   \renewcommand\customgenericname{#2}%
   \renewcommand\theinnercustomgeneric{##1}%
   \innercustomgeneric
  }
  {\endinnercustomgeneric}
}

\newcustomtheorem{customthm}{Theorem}
\newcustomtheorem{customlemma}{Lemma}

\newcommand\numberthis{\addtocounter{equation}{1}\tag{\theequation}}

\usepackage{xcolor}
\definecolor{darkblue}{rgb}{0.1,0.1,0.8}
\definecolor{DarkGreen}{rgb}{0,0.6,0}
\definecolor{brickred}{rgb}{0.8, 0.25, 0.33}
\definecolor{britishracinggreen}{rgb}{0.0, 0.26, 0.15}
\definecolor{calpolypomonagreen}{rgb}{0.12, 0.3, 0.17}
\definecolor{ao(english)}{rgb}{0.0, 0.5, 0.0}
	\definecolor{cadmiumgreen}{rgb}{0.0, 0.42, 0.24}
\definecolor{burgundy}{rgb}{0.5, 0.0, 0.13}
\usepackage{etoolbox}

\providetoggle{Blue_revision}
\settoggle{Blue_revision}{true}

\providetoggle{Track}
\settoggle{Track}{true}
\newcommand{\addv}[3]{%
	\iftoggle{Track}{%
    	\IfEqCase{#1}{%
       	 	{a}{\ifthenelse{\equal{#2}{ON}}{{\color{cadmiumgreen}#3}}{#3}}%
        	{b}{\ifthenelse{\equal{#2}{ON}}{{\color{brickred}#3}}{#3}}%
       		{c}{\ifthenelse{\equal{#2}{ON}}{{\color{burgundy}#3}}{#3}}%
    	}[\PackageError{tree}{Undefined option to tree: #1}{}]%
	}{#3}%
}

 \usepackage{hyperref}
\usepackage{xcolor}

\hypersetup{
colorlinks=true, %
 pdfstartview={FitH},
   linkcolor=red,
   citecolor=blue,
   urlcolor={blue!80!black}
}

\usepackage{graphicx}

\newcounter{relctr} %
\everydisplay\expandafter{\the\everydisplay\setcounter{relctr}{0}} %

\AtBeginDocument{} %

\newcommand{\etal}{\textit{et al. }}

\global\long\def\RR{\mathbb{R}}
\global\long\def\CC{\mathbb{C}}
\global\long\def\NN{\mathbb{N}}

\global\long\def\ZZ{\mathbb{Z}}
\global\long\def\EE{\mathbb{E}}
\global\long\def\PP{\mathbb{P}}

\global\long\def\11{\mathbbm{1}}

\newcommand{\bfa}{\mathbf{a}}
\newcommand{\bfb}{\mathbf{b}}

\newcommand{\bfe}{\mathbf{e}}

\newcommand{\bfr}{\mathbf{r}}
\newcommand{\bfs}{\mathbf{s}}
\newcommand{\bft}{\mathbf{t}}
\newcommand{\bfu}{\mathbf{u}}
\newcommand{\bfv}{\mathbf{v}}

\newcommand{\CA}{\mathcal{A}}

\newcommand{\CG}{\mathcal{G}}

\newcommand{\CM}{\mathcal{M}}

\newcommand{\CP}{\mathcal{P}}

\newcommand{\CS}{\mathcal{S}}

\newcommand{\CU}{\mathcal{U}}

\global\long\def\+{\oplus}

\newcommand{\prob}[1]{\PP\Big\{  #1 \Big\} }
\def\<{\langle}
\def\>{\rangle}

\ifdefined \var 
  \renewcommand{\var}{\mathsf{var}}
\else
  \newcommand{\var}{\mathsf{var}}
\fi

\ifdefined \abs \else
 \newcommand{\abs}[1]{\lvert#1\rvert}
\fi

\ifdefined \norm \else
 \newcommand{\norm}[1]{\lVert#1\rVert}
\fi

\ifdefined \set 
  \renewcommand{\set}[1]{\left\{#1\right\}}
\else
  \newcommand{\set}[1]{\left\{#1\right\}}
\fi

\ifdefined \poly 

\else
  \newcommand{\poly}{\optfont{poly}}
\fi

\usepackage{stackengine}

\usepackage[normalem]{ulem}

\DeclareMathOperator*{\tensor}{\otimes}

\providecommand{\tr}{tr}

\ifdefined \Tr 
  \renewcommand{\Tr}[1]{\tr \Big\{#1\Big\}}
\else
  \newcommand{\Tr}[1]{\tr \Big\{#1\Big\}}
\fi

 \def\id{I_d}

\def\Loss{\mathcal{L}}
\def\parameter{\overrightarrow{a}}
\def\dimH{\text{dim}_{\mathcal{H}}}

\newcommand{\optfont}[1]{\mathsf{#1}}
\def\qps{\sigma^\bfs}
\def\qpr{\sigma^\bfr}
\def\qas{a_{\bfs}}
\def\qar{a_{\bfr}}
\newcommand\ordersign[1]{\delta(#1)}
\newcommand{\ad}{\optfont{ad}}
\newcommand{\GL}{\optfont{GL}}

\newcommand{\g}{\optfont{g}}
\def\su{\mathfrak{su}}

\usepackage[nolist,nohyperlinks]{acronym}

\newacro{ptp}[PtP]{Point-to-Point}
\newacro{iid}[i.i.d.]{independent and identically distributed} 
\newacro{IID}[i.i.d.]{independent and identically distributed} 
\newacro{PAC}[PAC]{\textit{probably approximately correct}}
\newacro{VC}[VC]{Vapnik–Chervonenkis}
\newacro{ERM}[ERM]{\textit{empirical risk minimization}}
\newacro{SVM}[SVM]{support-vector machine}
\newacro{SGD}[SGD]{stochastic gradient descent}

\newacro{POVM}[POVM]{positive operator-valued measure}
\newacro{QPAC}[QPAC]{\textit{quantum probably approximately correct}}

\newacro{QSRM}[QSRM]{\textit{quantum shadow risk minimization}}
\newacro{QC}[QC]{quantum computer}
\newacro{ML}[ML]{machine learning}
\newacro{QML}[QML]{quantum-enhanced machine learning}
\newacro{NISQ}[NISQ]{noisy intermediate-scale quantum}
\newacro{VQA}[VQA]{variational quantum algorithm}
\newacro{VQE}[VQE]{variational quantum eigensolver}
\newacro{QAOA}[QAOA]{quantum approximate optimization algorithm}
\newacro{PQC}[PQC]{parameterized quantum circuit}
\newacro{PSR}[PSR]{parameter shift rule}
\newacro{HVA}[HVA]{Hamiltonian variational ansatz}
\newacro{DLA}[DLA]{dynamical Lie algebra}

\newacro{QNN}[QNN]{quantum neural network}
\newacro{QSGD}[QSGD]{quantum stochastic gradient descent}
\newacro{RQSGD}[RQSGD]{randomized quantum stochastic gradient descent}
\newacro{QGD}[QGD]{quantum gradient descent}
\newacro{QSS}[QSS]{quantum shadow sampling}
\newacro{CST}[CST]{classical shadow tomography}

\usepackage{booktabs}
\def\STATE{\State}

\settoggle{Track}{true}
\settoggle{Blue_revision}{true}

\begin{document}

\title{Efficient Gradient Estimation for Parameterized Quantum Systems with Lie Algebraic Symmetries}
\author{Mohsen Heidari}
\affiliation{Department of Computer Sciences, Indiana University, Bloomington, IN, USA}
\email{mheidar@iu.edu}
\homepage{https://homes.luddy.indiana.edu/mheidar/}
\thanks{\\
Portions of this work have been published in the 2025 Conference on Neural Information Processing Systems (NeurIPS 2025) \citep{Heidari2025a}. This manuscript substantially extends the previous publication by adding new results on general Hamiltonians (see Section \ref{sec:Extension}) and added omitted proofs of previous claims.}
\orcid{0000-0002-0012-2900}
\author{Masih Mozakka}
\affiliation{Department of Computer Sciences, Indiana University, Bloomington, IN, USA}
\author{Wojciech Szpankowski}
\affiliation{Department of Computer Sciences, Purdue University,    West Lafayette, IN, USA}

\maketitle

\begin{abstract}
	Gradient estimation is a central challenge in training parameterized quan-
tum circuits (PQCs) for hybrid quantum-classical optimization and learning problems. This difficulty arises from several factors,  including the exponential dimensionality of the Hilbert spaces and the information loss in quantum measurements.  Existing estimators, such as finite difference and the parameter shift rule, often fail to adequately address these challenges for certain classes of PQCs. In this work, we propose a novel gradient estimation framework that leverages the underlying Lie algebraic structure of PQCs, combined with the Hadamard test. By analyzing the differential of the matrix exponential, we derive an expression for the gradient as a linear combination of expectation values obtained via Hadamard tests. The coefficients in this decomposition depend solely on the circuit's parameterization and  can be estimated using state-of-the-art shadow tomography techniques. Hence, our approach enables efficient gradient estimation, requiring a number of measurement shots that scales logarithmically with the number of parameters, and with polynomial classical and quantum time. This is an exponential reduction in the measurement cost and a polynomial speed-up in time compared to existing works.
\end{abstract}

\section{Introduction}
Hybrid quantum-classical strategies have emerged as a leading approach for quantum optimization and learning  \citep{Benedetti2019,Cerezo2021}, and have been extensively studied across a broad range of domains, including optimization \citep{Farhi2014}, quantum chemistry \citep{Jones2019,Anselmetti2021,Grimsley2019,Delgado2021}, and quantum machine learning from classical and quantum data \citep{Farhi2018,Schuld2019,Mitarai2018,Liu2018,Havlicek2019,HeidariQuantum2021,Heidari2024,Heidari2023Quantum1,Huang2021}. \Ac{VQA} particularly has been a promising paradigm for quantum learning and inference, where a \ac{PQC} (a.k.a ansatz) is trained in a classical-quantum loop.  Gradient-based training methods have gained significant attention in the literature \cite{Harrow2021,Sweke2020,Schuld2018,Farhi2014,Farhi2018,Schuld2019,Mitarai2018} and have demonstrated  advantages in convergence rates compared to gradient-free methods. 

However, estimation of the gradient can be computationally challenging due to several factors including the exponential dimensionality of the Hilbert spaces, the no-cloning, information loss of quantum measurements, and non-commutativity of Hamiltonian terms. Therefore, each gradient estimation can have an exponential sample complexity leading to a high overhead and hence a bottleneck for the scalability of gradient-based VQAs. 

Several approaches have been introduced to estimate the gradient \cite{Farhi2018,Mitarai2018,Sweke2020,HeidariAAAI2022,Harrow2021,Schuld2018,Mitarai2019,Wiersema2024,Sohail2024}; but they often yield suboptimal gradient circuits for certain PQCs. Methods based on finite differences evaluate the objective function in the neighborhood of the parameters. They can be applied to general \acp{PQC}, but suffer from a slow convergence rate \cite{Harrow2021}. The well-known \ac{PSR} \cite{Schuld2018,Mitarai2018} relies on the Hadamard test with Pauli operators to estimate the partial derivatives. The Hadamard test is an efficient method that directly measures the partial derivatives and does not have the numerical instability of indirect methods such as finite differences.
 However, \ac{PSR} with Hadamard test is restricted to ans\"atze with two distinct eigenvalues. It can be adapted for more complex circuits via backpropagation, but it comes with high computational costs in terms of gate decomposition. Other existing methods often apply to more general circuits but have high overhead due to the extensive use of the ansatz with repeated measurements, and exponential classical computation \cite{Banchi2021,Wierichs2022,Theis2023}.

Lie algebraic structures in \acp{PQC} have been increasingly important in analysis and design of hybrid quantum-classical strategies. Ans\"atze that have  dynamical Lie algebra (DLA) with  polynomial dimensionality may not exhibit any Barren plateaus, which are flat regions in the parameter landscape \cite{Cerezo2021a,Fontana2023,McClean2018}. Moreover,  Lie algebraic symmetries have be used for classical simulation of quantum models \cite{Goh2023}. In this work, we build upon the Lie algebraic characterizations and we develop an efficient gradient estimation for general circuits based on the Hadamard test followed by post-processing steps. With that, we enable efficient applicability of the Hadamard test to generic \acp{PQC} without the need to change the ansatz structure and with low classical overhead. 

\subsection{Summary of the Main Results}
We analytically derive an explicit expression for the gradient of generic \acp{PQC} in terms of the expectation values of the Hadamard tests corresponding to a certain set of Pauli strings. Then, we develop a gradient estimation method using a series of Hadamard tests at the output of the ansatz followed by classical post-processing techniques including \ac{CST} \cite{Huang2020}.  A generic \ac{PQC} on $n$ qubits can be represented as $U(\parameter)=e^{iA(\parameter)},$ where $A$ is the parameterized Hamiltonian with $\overrightarrow{a}=(a_1, \cdots, a_p)$ as the vector of parameters. Typically, the Hamiltonian is written in term of Pauli strings as $A(\parameter) = \sum_{i=1}^p a_i P_i$, where $a_i\in \RR$ and $P_i$'s are tensor products of  Pauli operators. This formulation includes multi-layered \acp{PQC} (e.g., the hardware-efficient ansatz) and appears in a wide range of setups including \ac{QAOA}, many-body quantum systems (e.g., Ising model),  and  adiabatic evolutions. Typically, the objective is  minimizing a loss function $\Loss(\parameter)$ depending on the parameterization of the PQC, the input state, and the measurement observable. Hence, a gradient-based optimization can be  used given an estimation of $\nabla\Loss$.

Our gradient estimation method is based on a binary encoding of Pauli strings to capture the structural
properties of their commutation relations. Similar techniques have been used before in the context of the stabilizer formulation in quantum error correction \cite{Calderbank1996,Gottesman1997}. Then, we make a connection between this binary encoding and the differential of the matrix exponential map, studied in Lie algebra \cite{Rossmann2006}. We write the partial derivatives of the PQC   as an infinite-length linear combination of expectation values of  Hadamard tests for  various Pauli strings.  We show that when the Pauli strings  in  $A(\parameter)$ are closed under the commutation, the  infinite-length linear combination collapses to a finite number of terms that can be computed efficiently using the binary encoding. Such terms are written as the expectation value of a set of observables that can be estimated using CST, as an estimation procedure to estimate several observables with  sample complexity  growing logarithmically with the number of observables \cite{Huang2020}.   

The closedness condition means that for any pair of Pauli strings $P_i$ and $P_j$ appearing in $A(\parameter)$, their commutator $[P_i, P_j]$ can be written as a linear combination of the same set of Pauli strings. This condition implies that the Hamiltonian $A(\parameter)$ generates a \ac{DLA} with a basis consisting solely of the Pauli strings present in $A(\parameter)$ \cite{Schirmer2002,Wiersema2023}.  When the closedness condition is not directly satisfied, one can consider the sub algebra generated by $P_i$ terms in $A(\parameter)$ and apply the proposed gradient estimation method. In that case, the sample complexity and running time depend on the dimensionality of this sub-algebra. Therefore, the proposed estimation is efficient when the dimensionality is polynomial with $n$.   %

We also extend our results to a more general Hamiltonian structure where $A(\parameter)$ is decomposed into generic Hermitian terms  that correspond to certain physical interactions. We show that when the Hamiltonian \ac{DLA} has  $\poly(n)$ dimensionality, then the gradient $\nabla\Loss$ can be estimated polynomially.

The polynomial size assumption is already satisfied for several well-known Hamiltonian models, including variants of $2$-local Ising model (e.g., the transverse-field) and Kitaev chain \cite{Wiersema2023} used to model molecular dynamics. %
In addition, the polynomial dimensionality is an essential component to avoid the barren plateaus \cite{Larocca2022,Fontana2023,McClean2018,Cerezo2021a}. %
This stems from the fact that the variance of the gradient is inversely proportional to the dimension of the DLA \cite{Fontana2023}.%

Nevertheless, there still is a curiosity to understand the gradient estimation when DLA dimensionality grows exponentially with $n$. For that, we present a class of efficient methods to approximate the gradient and analyze  the error rate as a function of the approximation terms. Particularly, we present a truncation technique and a randomized algorithm for gradient estimation.

A more specific summary of our contributions is give below:

\begin{itemize}
	\item Showing the gradient of the loss function can be written as  $\nabla\Loss = \overrightarrow{D} f(V),$
	      where $f(z)=\frac{1-e^z}{z}$, 
		  $\overrightarrow{D}$ is the vector of  the expectation values of Hadamard tests for a set of Pauli strings, and $V$ is a matrix constructed  based on the parameters $\parameter$ (see Theorem \ref{thm:gradient mtx V}).
	\item  An algorithm that estimates $\nabla \Loss(\parameter)$ using $\tilde{O}(p)$ Hadamard tests and $O(p^3+pn)$ classical time.
	\item When the shadow norm of the observable is bounded, the gradient can be estimated with $O(\log p)$ copies and  $\poly(n)$ time (see Section \ref{sec:shadow}).
	\item A master's theorem for the general case where the Pauli terms are not closed under commutation (see Theorem \ref{thm:loss derivative non product}). %
	\item Generalization to Hamiltonians with generic Hermitian terms when the DLA has polynomial dimensionality (see Theorem \ref{thm:DLA derivative}).
	\item Approximation methods for gradient estimation when the DLA dimensionality is exponential (see Section \ref{sec:approximations}).
\end{itemize}

\subsection{Comparison With Related Methods}
We consider three complexity measures: (1) sample complexity, (2) the classical post-processing time, and (3) the number of distinct circuits that need to be evaluated to obtain all the partial derivatives.  %
Table \ref{tab:compare} demonstrates a simplified comparison with existing works for gradient estimation.  For a more intuitive comparison, it is assumed that the PQC acts on $n$ qubits with gate complexity  $\Theta(n^b)$ and  $p=\Theta(n^a)$ parameters, where $a$ and $b$ are arbitrary constants. The table shows that our approach provides an exponential advantage in terms of the sample complexity and a polynomial speed-up in classical running time. Below, we highlight some of the most relevant approaches for comparison to our work.

\begin{table}[t]
	\centering
	\begin{tabular}{l|ccc}
		& Circuit Changes         & Sample Complexity        & Running Time*                \\
		\midrule
		SPSR\cite{Banchi2021,Wierichs2022}  & $p$                    & $O(p)$                 & $O(n^{a+b})$                 \\
		NPSR\cite{Theis2023}                & $\tilde{O}(p\norm{A})$ & $\tilde{O}(p\norm{A})$ & $O(n^{a+b})$                 \\
		$\mathit{SU}(N)$\cite{Wiersema2024} & $p$                    & $O(p)$                 & $\exp(\Theta(n))$            \\
		Backpropagation \cite{Abbas2023}					& -						 & $O(\log^2 p)$          & $p\exp\qty\big(\tilde{O}(n))$    \\
		Theorem \ref{thm:subgroup}, \ref{thm:shadow cliford}    & $ 0$                & $O(p), O(\log p)^\star$        & $\tilde{O}(n^b +n^{3a})$ \\
		\bottomrule
	\end{tabular}
	\caption{Rough comparison of various methods for estimating the gradient of a PQC with $p$ parameters and DLA based on Pauli strings with $\poly(n)$ dimensionality. * For presentation convenience of the runtime, it is assumed that $p=\Theta(n^a)$ and that each use of the ansatz takes $\Theta(n^b)$ quantum time, where $a$ and $b$ are arbitrary constants. %
		$^\star$ Conditioned on bounded shadow norm of $O$ and corresponding polynomial time complexity.}
	\label{tab:compare}
\end{table}

\paragraph{Stochastic \ac{PSR}:}
This method is a generalization of \ac{PSR} \cite{Banchi2021,Wierichs2022}, where each partial derivative is written as an integral, and a Monte Carlo strategy is used to estimate it.   David \etal  \cite{Wierichs2022} also presented a generalization of the \ac{PSR} using the Discrete Fourier series. Here, the parameters are jointly shifted depending on the spectrum of the ansatz.

This method is efficient when the Hamiltonian $A$ is promised to have equidistant eigenvalues. For a generic PQC one first  needs to compute the spectral decomposition of $A$ to find the pattern of the parameter shifts. Therefore, this process in general takes $\exp{\Theta(n)}$ classical time as $A$ is an exponentially large matrix.

\paragraph{Nyquist \ac{PSR}:} Recently  \cite{Theis2023} proposed a  shift rule for \acp{PQC} where only the parameters are shifted without any other modifications of the ansatz. The method relies on a beautiful connection between the Nyquist-Shannon sampling theorem and the Fourier series that was observed earlier in \cite{Wierichs2022,Vidal2018}. The number of unique circuits for this estimation scales with $p$ and the difference between the maximum and minimum eigenvalues of $A$ --- a quantity bounded by the operator norm $\norm{A}$.   As the authors reported, this method has low approximation error when the parameter value $\theta$ is large enough. More precisely, the approximation error is $O(\frac{1}{c^2})$ as long as $\theta = (1 - \Omega(1))c$, where $c$ is the maximum magnitude of a parameter value.

\paragraph{Lie algebraic:}
This is another approach \cite{Wiersema2024} based on Lie algebra and a nice connection to the geometry of $SU(2^n)$ matrices and the adjoint operator. The gradient is calculated by finding the Jacobian matrix of the matrix representation of PQC. Hence, the running time scales as $p 2^{\Theta(n)}$. 

\paragraph{Shadow tomography:} A recent work \cite{Abbas2023} proposes a quantum backpropagation method for PQC of the form $U(\parameter) = \prod_{j=1}^p e^{ia_j P_j} U_j$, where $U_j$ are fixed unitaries and $P_j$ are fixed Pauli words. Leveraging the shadow tomography  of \citep{Aaronson2018}, the method achieves a sample complexity scaling as $\log^2 p$, but with the cost of an exponential classical memory requirement of $p2^{\tilde{O}(n)}$.

\paragraph{Classical simulation:}
Existing simulation methods such as g-sim \cite{Goh2023} can efficiently simulate
 a quantum system and hence compute the gradient, under the assumption that both the Hamiltonian $A(\parameter)$ and the observable $O$ have polynomial Lie dimensionality. In contrast, we only require $A(\parameter)$ (rather than both) to have polynomial Lie dimensionality. %

Assuming  both input state  $\rho$ and observable $O$ are efficiently classically simulatable, our work and that of \cite{Goh2023} have time complexity scaling polynomially with the dimension of the Lie algebra. Our work is complementary to \cite{Goh2023} and offers distinct advantages when   $\rho$ and/or $O$ are  not classically simulatable. This arises, for example, when $O$ does not lie in a Lie algebra of polynomial dimension, and  $\rho$ is a physical quantum state obtained via an external process (e.g., quantum sensing) or generated by a unitary that does not have polynomial Lie dimensionality. 

When $\rho$ is stored in a physical quantum state, \cite{Goh2023} relies on computing expectation values of the Lie algebra basis elements ($B_{\alpha}^{(\lambda)}$ in the reference) under $\rho$. This is only efficient when certain classes of $\rho$ including product states or stabilizer states. In  more general setting, such expectation values must be estimated in a quantum computer and the sample complexity of the estimation may not be polynomial in $n$. %
Furthermore, even when $\rho$ is classically simulatable, we expect significant (possibly exponential) separation from \cite{Goh2023} when the observable does not have polynomial Lie dimensionality while it can be implemented in polynomial quantum time  on a quantum device and has bounded shadow norm.
For example, in fidelity estimation, $O=\ketbra{\phi}{\phi}$, where $\ket{\phi}=e^{iH}\ket{0}$ and $H$ is a Hamiltonian not admitting a polynomial Lie decomposition. Here, $O$ is a low-rank observable with bounded shadow norm, but exponential Lie dimensionality. Hence, classical simulation runs exponentially in time.

\section{Preliminaries and Model}
\paragraph{Notation.}  For any $d\in \NN$, let $H_d$ be a the Hilbert space of $d$-qubits. %
By $\mathcal{B}(H_d)$ denote the space of all bounded (linear) operators acting on $H_d$.  The identity operator is denoted by $\id$. %
As usual, a quantum state, is denoted by a \textit{density operator}; that is a Hermitian, unit-trace, and non-negative linear operator.
A quantum measurement/observable is modeled as a \ac{POVM}. A POVM $\mathcal{M}$ is represented by a set of operators $\mathcal{M}:=\{M_v, v\in\mathcal{V}\}$, where $\mathcal{V}$ is the (finite) set of possible outcomes. The operators $M_v$ are non-negative and form a resolution of identity, that is $M_v=M_v^\dagger\geq 0,  \sum_v M_v =\id.$ %
The operator norm (infinity norm) for an operator $A$ is defined as $\norm{A}= \sup_{\ket{\phi}}\norm{A\ket{\phi}}.$
The Hilbert-Schmidt norm of $A$ is defined as $\norm{A}_2:=\sqrt{\tr{A^\dagger A}}$  
\subsection{Pauli Group}
Together with the identity, they  are denoted as $\{\sigma^0, \sigma^1, \sigma^2, \sigma^3\}$ with
\begin{align*}
	\sigma^0 = I= \begin{pmatrix}
		              1 & 0 \\ 0 & 1
	              \end{pmatrix}
	\qquad
	\sigma^1 =X = \begin{pmatrix}
		              0 & 1 \\ 1 & 0
	              \end{pmatrix},
	\quad  \sigma^2 = Y=  \begin{pmatrix}
		                      0 & -i \\ i & 0
	                      \end{pmatrix},
	\quad  \sigma^3 =Z =  \begin{pmatrix}
		                      1 & 0 \\ 0 & -1
	                      \end{pmatrix}.
\end{align*}
The single-qubit Pauli group $\CP_1$ is the 16 element set $\{c\sigma^s: s=0,1,2,3, ~ c=\pm 1, \pm i \}$.  %

For $n$ qubit systems, the Pauli tensor products are denoted as
$\sigma^{\bfs}:= \sigma^{s_1}\tensor \sigma^{s_2}\tensor \cdots \tensor \sigma^{s_d},$ for all  $\bfs\in \{0,1,2,3\}^n.$
The $n$-qubit Pauli group $\mathcal{P}_n$ is then defined as the group generated by $n$-fold tensor products of the Pauli matrices:
\[
	\mathcal{P}_n = \left\{ c\qps: \bfs\in \{0,1,2,3\}^n, ~  c=\pm 1, \pm i\right\}.
\]
This group has $4^{n+1}$ elements and  spans any operator on the space of $n$ qubits:
\begin{fact}
	Any bounded  operator $A$ on   $n$ qubits can be uniquely written as   $A = \sum_{\bfs \in \{0,1,2,3\}^n} a_{\bfs} ~\sigma^{\bfs},$ where $a_\bfs = \frac{1}{2^n}\tr\big\{A\qps\big\}.$
\end{fact}

\subsection{General Framework}
The objective is to minimize a cost function defined as
\begin{equation}\label{eq:VQA cost function}
	\Loss(\parameter) : = \tr{{O} ~ U(\overrightarrow{a}) \rho U(\overrightarrow{a})^\dagger},
\end{equation}
where  $O$ is a fixed observable, $\rho$ is the initial (mixed) state, and $U(\parameter)$ is a parameterized quantum circuit with $\overrightarrow{a}=(a_1, \cdots, a_p)$ as the vector of parameters. As $U$ is unitary, we can always write $U(\parameter) = e^{iA(\parameter)}$ for some Hamiltonian matrix $A$. %
To ensure computational tractability, it is assumed that the number of parameters $p=\poly(n),$ with $n$ being the number of qubits. %

Making iterative progress in the direction of the steepest descent is one of the most popular optimization techniques in \acp{VQA}, as it has been in classical problems. Ideally, a gradient descent optimizer applies the following update rule at each iteration $t$:
\begin{align}\label{eq:QSGD ideal}
	\overrightarrow{a}^{(t+1)} = \overrightarrow{a}^{(t)} - \eta_t \nabla \Loss(\overrightarrow{a}^{(t)}),
\end{align}
where $\eta_t \in \RR$ is the learning rate at iteration $t$.
The above update rule is not realistic as the objective function $\Loss(\overrightarrow{a})$ is an expectation value, and the characteristics of $\rho$ are either unknown or computationally intractable. %

There have been several approaches to estimating the gradient \cite{Farhi2018,Mitarai2018,Sweke2020,HeidariAAAI2022,Harrow2021,Schuld2018,Mitarai2019,Wiersema2024,Sohail2024,HeidariAAAI2022}. The zeroth-order approach (e.g., finite differences) evaluates the objective function in the neighborhood of the parameters.  Although it is a generic approach, recent studies showed their drawbacks in terms of convergence rate \cite{Harrow2021}.  First-order methods (e.g., parameter shift rule)  directly calculate the partial derivatives \cite{Schuld2018}. 

\subsection{Hadamard Test}
Given a unitary $U$, the Hadamard test, which is a special case of the phase estimation, is a quantum circuit that we can use to estimate the real or imaginary value of $\expval{U}{\psi}$ for some state $\ket{\psi}$. The circuit consists of two Hadamard gates and a controlled version of $U$.

When the ansatz has a simple form $U(\theta) = e^{i\theta \qps}$ its derivative can be written as below and be directly estimated via a Hadamard test \cite{Mitarai2018}:
\begin{equation}\label{eq:D_i}
	D_\bfs := i\tr{O[\qps, \rho^{out}]},
\end{equation}
where $\rho^{out}=U(\parameter)\rho U(\parameter)^\dagger$ is the output of the ansatz on input $\rho$. The above equation is based on the fact that $\dv{U}{\theta} = i\qps U(\theta)$. Estimating the gradient through the Hadamard test can enhance   computing efficiency, and allows for the use of measurement optimization techniques. Moreover, it can be used to compute higher-order partial derivatives used in higher-order optimization algorithms \cite{Li2024}. 

This approach can be extended to product ans\"atze that are concatenation of single parametric or non-parametric unitary circuits of the form
\begin{align*}
	U(\parameter) = U_L(a_L)V_L \cdots U_1(a_1)V_1,
\end{align*}
where $V_l$ is non-parametric and each parametric layer is of the form
$	U_l(a_l):= e^{i a_{\bfs_l} \sigma^{\bfs_l}}$
with $\sigma^{\bfs_l}$ being a Pauli string. %
More formally, the partial derivative of a product ans\"atze can be written in terms of commutators with associated Pauli strings. This statement is summarized below:
\begin{fact}\label{fact:loss derivative}
	Let ${\rho}_l^{out} =U_{\leq l}\ketbra{\phi} U_{\leq l}^\dagger$ denote the density operator of the output state at layer $l$. Then, the derivative of the loss is given by:
	\begin{align}\label{eq:losss derivative}
		\frac{\partial  \Loss%
		}{\partial a_{\bfs_l}} =   \tr{O U_{>l} \big[ \sigma^{\bfs_l}, {\rho}_l^{out}\big]U^\dagger_{>l}},
	\end{align}
	where $[\cdot, \cdot]$ is the commutator operation, and $U_{>l}$ denotes the unitary components from layer $l+1$ to $L$.
\end{fact}
The expression above can be implemented using a quantum circuit with measurements at the end. The circuit is shown in Fig.~\ref{fig:loss derivative circuit multi} which is a special example of the generalized Hadamard test \cite{Mitarai2019,Harrow2021}.  In this procedure, given a sample $\ket{\phi}$ and an ancilla qubit $\ket{0}$, we first apply the first $l$ layers of the ansatz and then apply a special circuit with controlled rotations for taking the derivative of the loss. Then, the rest of the layers of the ansatz are applied and measurement is performed. Extending the Hadamard test for gradient estimation of more generic ans\"atze is one of the main focuses of this work.

\begin{figure}
	\centering
	\resizebox{0.8\textwidth}{!}{
		\begin{quantikz}[scale=0.2]
			\lstick{$\ket{\phi}$}&\gate{U_{\leq l}} &\gate{R_{\bfs_l}(+\frac{\pi}{2})}&\qw  &\gate{R_{\bfs_l}(-\frac{\pi}{2})} &\qw & \gate{U_{>l}} & \meter{} \arrow[r]&\rstick{$\hat{y}$}\\
			\lstick{$\ket{0}$} &\gate{H} &\ctrl{-1}  &\gate{X}& \ctrl{-1}  & \gate{X}& \qw & \meter{} \arrow[r]&\rstick{z}
		\end{quantikz}
	}
	\caption{Hadamard test with backpropagation for measuring the partial derivative of product ans\"atze with respect to a parameter $a_{\bfs_l}$ appearing at layer $l$. Here $U_{\leq l}$ corresponds to the first $l$ layers of the ansatz, and $U_{>l}$ to the remaining layers. in addition,  $X$ is the X-gate, $H$ is the Hadamard gate, and $R_{\bfs_l}$ is the controlled rotation around Pauli $\sigma^{\bfs_l}$. %
	}
	\label{fig:loss derivative circuit multi}
\end{figure}

\subsection{Differential of The Matrix Exponential}
The matrix exponential is defined as
\[
	\exp(X) = e^X = \sum_{k=0}^{\infty} \frac{X^k}{k!},
\]
where $X$ is a square matrix.  Due to non-commutativity of matrix product, the differential of the exponential map has a more complex formula compared to the exponential function. Suppose $X(\tau)$ is a differentiable matrix (linear operator) as a function of the variable $\tau\in \RR$. The \text{adjoint} map is defined as the mapping $$\ad_X(Y)= [X,Y]=XY-YX$$ for square $n\times n$ matrixes $X,Y\in \GL(n, \CC)$.  Then, for any $k=0,1, ...$, we can define $$\ad_X^k(Y) = [X, \cdots, [X,Y]\cdots ].$$

The exponential map and the adjoint are  fundamental concepts in the theory of Lie groups and Lie algebras, describing how a Lie group or Lie algebra acts on its own Lie algebra by conjugation (a standard text book on this topic is \cite{Rossmann2006}).  The adjoint operator is connected to the derivative of the matrix exponential.

\begin{theorem}[\cite{Rossmann2006}]\label{thm:exp differential}
	Suppose $X(\tau)$ is a differentiable  (linear) operator with respect to a variable $\tau \in \RR$. Then, the differential of the  matrix exponential is given by
	\begin{align}\label{eq:exp differential}
		\dv{\exp{X(\tau)}}{\tau} = \exp{X(\tau)} \frac{1-\exp{-\ad_X}}{\ad_X}\dv{X(\tau)}{\tau}.
	\end{align}
\end{theorem}

\section{Binary Encoding of Pauli Operators}\label{sec:binary encoding}
Representation of Pauli words with binary strings has been studied in  quantum error correction \cite{Calderbank1996,Gottesman1997} to construct the Stabilizer codes.  It is well-known that the Pauli group  $\mathcal{P}_n$ is \textit{isomorphic} to the \textit{semi direct product} of $\ZZ_4$ and $\ZZ_2^{2n}$.   This binary representation allows us to write the partial derivatives of $\Loss(\parameter)$ as linear combination of terms related to the  Hadamard tests applied to the ansatz output.  In this work, we present an explicit form of such a mapping. Below, we start with a series of notations and definitions and present a binary encoding technique to study the products of various Pauli string terms.

The phase scalar $c\in \{\pm 1, \pm i\}$ in $\CP_n$ can be written as $i^{a},$ where  $a\in \ZZ_4$, the modulo-four group. As for the Pauli operators, consider the binary vector group $\ZZ_2 \times \ZZ_2 = \set{(0|0), (0|1), (1|0), (1|1)}$ with the element-wise modulo two addition:
\begin{align*}
	(a^0|a^1)+(b^0|b^1) = (a^0\oplus b^0 | a^1\oplus b^1),
\end{align*}
where $\oplus$ is the binary addition. We use  $(\cdot |\cdot)$ to distinguish between the first and the second components of elements of $\ZZ_2\times \ZZ_2$.  We associate the identity and each Pauli operator with elements of $\ZZ_2 \times \ZZ_2$ as
\begin{align*}
	\sigma^0 \rightarrow (0|0), \qquad \sigma^1 \rightarrow (0|1), \qquad \sigma^2 \rightarrow(1|0), \qquad \sigma^3 \rightarrow(1|1).
\end{align*}

Extending to $n$-qubits,  $(\ZZ_2 \times \ZZ_2)^n$ is defined as the set of all $(\bfa^0|\bfa^1)$ for binary strings $\bfa^0, \bfa^1\in \ZZ_2^n$ with the element-wise  addition:
\begin{align*}
	(\bfa^0|\bfa^1)\oplus (\bfb^0|\bfb^1) = (\bfa^0\oplus \bfb^0 | \bfa^1\oplus \bfb^1).
\end{align*}
Similarly, any Pauli string $\qps, \bfs \in \{0,1,2,3\}^n$ is associated with $(\bfs^0|\bfs^1)$ where $\bfs^0= (s^0_1, \cdots, s^0_d), \bfs^1= (s^1_1, \cdots, s^1_d)$ are binary strings. Therefore, we frequently switch between representations of s as a member of $\ZZ_4^n$ and $(\ZZ_2\times \ZZ_2)^n$.

\begin{example}
	The Pauli string $X\tensor Y\tensor X$ which is associated with $\qps$ with $\bfs = (1,2,1)$ has the following binary encoding: $\qty((0,1,0) | (1, 0, 1)).$
\end{example}
 The above association induces an isomorphism between the Pauli group $\CP_n$ and the semi direct product of $\ZZ_4$ and $\ZZ_2^{2n}$:
\begin{remark}
	$\mathcal{P}_n$ is isomorphic to the \textit{central product} of $\ZZ_4$ and $\ZZ_2^{2n}$.  Ignoring the phase scalar $c$ in $\CP_n$, the group quotient  is \textit{isomorphic} to the Binary vector group  $\ZZ_2^{2n}$. On the other hand, the scalars $\{\pm 1, \pm i\}$ are \textit{isomorphic} to the cyclic group of order 4 which is $\ZZ_4=\{0,1,2,3\}$ with modulo-4 addition as the group action.  
\end{remark}

We formalize the above mapping in the following definition.
\begin{definition}\label{def:Pauli phi map}
	Any element of $\CP_n$, written as $i^a\qps$, is represented as $(a, \bfs)$ where  $a\in \ZZ_4$ and $\bfs \in \ZZ_2^{2n}$. Such a representation is defined by the mapping  $\phi: \CP_n\rightarrow \ZZ_4\times \ZZ_2^{2n}$, that sends $\phi(i^a\qps)=(a,\bfs)$.
\end{definition}

\paragraph{Encoding the products of Pauli words.}
Inspired by the Levi-Civita symbol, we define a sign function on $i,j\in \ZZ_4$ (or $\ZZ_2\times \ZZ_2$) as
\begin{align*}
	\ordersign{i,j}:=\begin{cases} -1 & \text{if}~ (i,j) = (1,3), (2,1), (3,2) \\
              1  & \text{if}~ (i,j) = (1,2), (2,3), (3,1) \\
              0  & \text{otherwise}.                       \\\end{cases}\numberthis\label{eq:ordersign}
\end{align*}
For vectors $\bfu,\bfv$, define $\ordersign{\bfu,\bfv}=\sum_j \ordersign{u_j,v_j}$.

\begin{lemma}\label{lem:product of Paulies}
	The product of any pair of Pauli strings $\qps, \qpr$ equals
	$$\qps\qpr = i^{\ordersign{\bfs, \bfr}} \sigma^{\bfs \oplus \bfr}.$$
\end{lemma}
\begin{proof}
	It is not difficult to verify the lemma for single qubit case by testing it for all possible combinations. Building on the single qubit case, for general $n>1$, with the tensor product, we have that
	\begin{align*}
		\qps \qpr  &= \bigotimes_j \sigma^{s_j}\sigma^{r_j}\\
		 &=  \bigotimes_j i^{\ordersign{s_j, r_j}} \sigma^{s_j \oplus r_j}\\
		 & = i^{\sum_j \ordersign{s_j,r_j}}  \bigotimes_j  \sigma^{s_j \oplus r_j}\\
		 & = i^{\sum_j \ordersign{s_j,r_j}} \sigma^{\bfs \oplus \bfr}.
	\end{align*}
	where the second equality is due to the single qubit case and the third equality is due to the fact that the phase scalars are multiplicative. The last equality is due to the definition of $\oplus$ for binary strings.
\end{proof}

\begin{lemma}
    The mapping $\phi: \CP_n\rightarrow \ZZ_4 \rtimes \ZZ_2^{2n}$  defined as $c\qps\mapsto \phi(c\qpr)=(c,\bfr)$ is an isomorphism of $\CP_n$ onto the semi direct  product of $\ZZ_4$ and $\ZZ_2^{2n}$. 
\end{lemma}
\begin{proof}
    We define the semi direct product of $\ZZ_4$ and $\ZZ_2^{2n}$ by defining the action that maps $\ZZ_4$ to the group of automorphisms of $\ZZ_2^{2n}$. For $a\in \ZZ_4$, we define the automorphism $\alpha_a: \ZZ_2^{2n}\rightarrow \ZZ_2^{2n}$ as $\alpha_a(\bfs) = a\bfs$ where $a\bfs$ is defined as the element-wise product of $a$ and $\bfs$. Then, the semi direct product of $\ZZ_4$ and $\ZZ_2^{2n}$ is defined as the mapping that sends any pair of elements $(a, \bfs), (b, \bfr)$ to
    $(mod_4(a+b+\sum_j \ordersign{s_j, r_j}), \bfs\oplus\bfr)$.
    With this definition, Lemma \ref{lem:product of Paulies} shows that $\phi(\qps\qpr)=\phi(\qps)\phi(\qpr)$. Adding the phase scalars is addressed by this definition. 
\end{proof}
With this result, and by adding the phase scalars, for any pair $i^a\qps$ and $i^b\qpr$ from the Pauli group $\CP_n$, the product is characterized by the $\phi$ map in Definition \ref{def:Pauli phi map} as
$$(i^a\qps) (i^b\qpr) = \phi^{-1}\qty(\qty(\text{mod}_4(a+b+\ordersign{\bfs, \bfr}), \bfs\oplus\bfr)).$$

\begin{example}
	Consider $P_1 = i X\tensor Y\tensor X$ and $P_2 = Z\tensor X\tensor Y$ that are encoded to $\qty(1, ( (010) | (101)))$ and $\qty(0, ((101) | (110)))$, respectively.  Then, $P_1P_2$ is associated  with the binary encoding $\qty(0, ((111) | (011)))$ which represents  $Y\tensor Z\tensor Z$.
\end{example}

\begin{example}
    For $n=1$ as an example, the  Cayley Table of product of Pauli matrices is given by the following table:
    \[
    \begin{array}{c|ccc}
    \bullet & (b, 01) & (b, 10) & (b, 11) \\
    \hline
    (a, 01) & ((a + b)_4, 00) & ((a + b+1)_4, 11) & (((a + b-1)_4), 10) \\
    (a, 10) & ((a + b-1)_4, 11)  & ((a + b)_4, 00)  & ((a + b+1)_4, 01)  \\
    (a, 11) & ((a + b+1)_4, 10)  & ((a + b-1)_4, 01)  & ((a + b)_4, 00)  \\
    \end{array}
    \]
\end{example}

The binary representation also characterizes the commutation relations between Pauli strings.
\begin{lemma}\label{lem:bracket and K_4}
	The commutator of any pair of Pauli strings $\qps, \sigma^{\bfr}$ is given by
	$[\qps, \sigma^{\bfr} ] = 2 i^{\ordersign{\bfs, \bfr}} \sigma^{\bfs \oplus \bfr}.$

\end{lemma}
\begin{proof}
	The proof is based on Lemma \ref{lem:product of Paulies}:
	\begin{align*}
		[\qps, \sigma^{\bfr} ] &= \qps \sigma^{\bfr} - \sigma^{\bfr} \qps \\
							   & = i^{\sum_j \ordersign{s_j,r_j}} \sigma^{\bfs \oplus \bfr} - i^{\sum_j \ordersign{r_j,s_j}} \sigma^{\bfr \oplus \bfs}.
	\end{align*}
	Note that $\bfr \oplus \bfs = \bfs \oplus \bfr$ and that $\ordersign{r_j,s_j}=-\ordersign{s_j,j_j}$. Therefore, as $i^{-1}=-i$, the commutator equals to   $[\qps, \sigma^{\bfr} ] = 2 i^{\sum_j \ordersign{s_j,r_j}} \sigma^{\bfs \oplus \bfr}$.
\end{proof}

\section{Gradient Formulation for PQC with Pauli Decomposition}
Recall that the objective function $\Loss(\parameter)$ in \eqref{eq:VQA cost function} and the expectation values $D_\bfs$ in \eqref{eq:D_i}  can be estimated using Hadamard test. We show how the gradient of $\Loss$ can be evaluated in terms of $D_\bfs$ quantities. 
We present a master theorem writing the partial derivative is a linear combination of several terms similar to the product ans\"atze as in Fact \ref{fact:loss derivative}. %
Then we show that under certain algebraic  structures, the linear combination collapses into $\poly(p)$ terms. %

\begin{theorem}\label{thm:loss derivative non product}
	Consider an ansatz of the form $U(\overrightarrow{a})=\exp{i A(\parameter)}$, where $A(\parameter)=\sum_{\bfs\in \mathcal{S}} \qas \qps$ for some $\mathcal{S}\subseteq \set{0,1,2,3}^n$. %
	Then, for any $\bfr\in \mathcal{S}$, 
	\begin{align}\label{eq:pdv infinite sum}
		\frac{\partial \Loss(\parameter)}{\partial \qar} = \sum_{k=0}^\infty  \frac{(2i)^k}{(k+1)!} \sum_{\bfs_1 \in \mathcal{S}} \cdots \sum_{\bfs_k \in \mathcal{S}}  \qty\Big(\prod_{j=1}^k  a_{\bfs_j}  i^{\ordersign{\bfs_j, \bfs_1 \oplus \cdots \oplus \bfs_{j-1} \oplus \bfr}})  D_{\bfs_1 \oplus \cdots \oplus \bfs_k \oplus \bfr},
	\end{align}
	where $D_{\bft}$ is defined as in \eqref{eq:D_i},  and $\ordersign{\cdot, \cdot}$ is defined as in \eqref{eq:ordersign}. 
\end{theorem}
The proof of the theorem is delayed until Section \ref{sec:proof thm derivative}. With this theorem, the derivative of loss can be measured by applying the standard Hadamard test followed by classical corrections. As a sanity check, Appendix \ref{sec:example single qubit} presents an explicit derivation of the gradient in the single qubit case.

Next, we study the simplification of the  partial derivative for Hamiltonians that have an algebraic structure with the goal to find an efficient estimator of  the gradient.

We show that when $\mathcal{S}$ in the theorem is isomorphism to a subgroup of $(\ZZ_2)^{2n}$, the expression simplifies significantly. In that case,the set of Pauli strings $\qps, \bfs\in \CS$ is closed under the commutation up to a global phase.  %
\begin{figure}[ht]
	\centering
	\begin{tikzpicture}[
	scale=0.75,
  transform shape,
  font =\Large,
  x=1cm,y=1cm,
  >=Stealth,
  line cap=round,
  line join=round,
  every node/.style={inner sep=1pt}
]
  \coordinate (L) at (0,0);
  \draw[thick,->] (L) -- ++(-1.2,-1.6) node[below left] {$a_1$};
  \draw[thick,->] (L) -- ++(3.4,-0.6)  node[below right] {$a_2$};
  \draw[thick,->] (L) -- ++(0,3.5)      node[above] {$a_3$};

  \draw[red,thick,->] (L) -- ++(0.9,1.6)
    node[above right,xshift=-2em] {$\nabla \Loss=(\pdv{\Loss}{a_1}, \pdv{\Loss}{a_2}, \pdv{\Loss}{a_3})$};

  \coordinate (R) at (11.8,-0.15);
  \draw[thick,->] (R) -- ++(-1.2,-1.7) node[below left] {$\bfe_1$};
  \draw[thick,->] (R) -- ++(3.2,-0.55) node[below right] {$\bfe_2$};
  \draw[thick,->] (R) -- ++(0,3.45)    node[above] {$\bfe_3$};

  \draw[red,thick,->] (R) -- ++(1.5,1.55)
    node[above right] {$\vec{D}=(D_{\bfs_1}, D_{\bfs_2}, D_{\bfs_3})$};

  \draw[very thick,<-]
    (5,1.3) .. controls (6.5,2.35) and (7.9,2.35) .. (10,1.3);

  \node at (7.5,2.55) {\Large $\nabla\Loss = \overrightarrow{D}f(V)$};
\end{tikzpicture}
	\caption{
		The left picture represents the parameter space of the ansatz with three parameters $\parameter = (a_1, a_2, a_3)$ corresponding to the Pauli strings $\sigma^{\bfs_1}, \sigma^{\bfs_2}, \sigma^{\bfs_3}$ for $\bfs_1, \bfs_2, \bfs_3 \in \CS$. Here, the gradient is a vector in this space. When the corresponding Pauli strings are closed under commutation, the gradient can be expressed as a vector in the space of Hadamard test outcomes $D_\bfs$ as shown in the right picture. The transformation from the right to the left picture is given by the matrix $V$ defined in \eqref{eq:V matrix column}. 
		}
	\label{fig:matrix illustration}
\end{figure}

For that, we first consider a geometric representation of the gradient in terms of $D_\bfs$ terms. Let $\CS = \{\bfs_1, \cdots, \bfs_p\}$, where $p$ is the number of parameters.  Let $\bfe_1 = (1, 0,\cdots, 0)^T, \cdots, \bfe_p =(0, \cdots, 0, 1)^T$ be the canonical basis vectors in $\RR^p$. We associate each $\bfs_j$ with $\bfe_j$ which is also denoted by $\bfe_{(\bfs_j)}$. Now consider the vector of  the expectation terms $\overrightarrow{D} = (D_{\bfs_1}, \cdots, D_{\bfs_p})$ as a row  vector in $\RR^p$:
\[
	\overrightarrow{D} = (D_{\bfs_1}, \cdots, D_{\bfs_p}) = \sum_{j=1}^p D_{\bfs_j} \bfe_{(\bfs_j)}.
\]

The following theorem shows that the gradient can be expressed as a vector derived from $\overrightarrow{D}$ by applying a matrix transformation. The matrix is defined in terms of the parameters and the group structure of the Pauli strings.

\begin{theorem}\label{thm:gradient mtx V}\label{thm:subgroup}
	Consider the ansatz $U(\overrightarrow{a})=\exp{i A(\parameter)}$, with $A(\parameter)=\sum_{\bfs\in {\CS}} \qas \qps$, where $\CS$ is closed under commutation and $|\CS|=p$. Define a $p\times p$ matrix $V$ where the $j$th column is given by
	\begin{align}\label{eq:V matrix column}
		\bfv_j := 2i \sum_{\bfs \in \mathcal{S}}   a_{\bfs}  i^{\ordersign{\bfs, \bfs_j}}~   \bfe_{(\bfs \oplus \bfs_j)}.
	\end{align}
	Then, the gradient of the loss is given by 
	\begin{align}\label{eq:gradient mtx V}
		\nabla\Loss = \overrightarrow{D}f(V),
	\end{align}
	where $f(z)=\frac{1-e^z}{z}$ with $f(0)=1$.
	Moreover there is an algorithm (Algorithm \ref{alg:subgroup}) that computes $\nabla\Loss$ in $O(p^3 + pn)$ time with $O(p)$ use  of the Hadamard tests.
\end{theorem}

\begin{proof}
First, as the Paulies  $\qps, \bfs \in \CS$ are  closed under commutation, $\CS$ is closed under the ``$\oplus$" operation. Hence, each term $\bfs_1\oplus \cdots \oplus \bfs_k \oplus \bfr$  in \eqref{eq:pdv infinite sum} remains in $\CS$. As a result the infinite-length sum  in Theorem \ref{thm:loss derivative non product} reduces to a linear combination of their form
\begin{align*}
	\frac{\partial \Loss}{\partial a_{\bfr}} = \sum_{\bfs\in \CS}  g_\bfs(\bfr)  D_{\bfs},
\end{align*}
where $g_\bfs(\bfr)\in \RR$. This means that the partial derivative is always a linear combination of $D_\bfs, \bfs\in \CS$ terms. However, the challenge is in computing the coefficients $g_\bfs(\bfr)$ that are coming from an infinite-length sum. We present a method to address this issue. 

Recall the association of $\bfs$ with the canonical basis vectors in $\RR^p$ and the vector $\overrightarrow{D}$ in $\RR^p$. Then $\pdv{\Loss}{\qar}$ is a vector in $\RR^p$ with the representation
\begin{align*}
	\pdv{\Loss}{\qar} \equiv (g_{\bfs_1}(\bfr), \cdots, g_{\bfs_p}(\bfr))\in \RR^p.
\end{align*}
Recall the matrix $V$. By an induction argument, it is not difficult to check that 
	\begin{align*}
		V^k\bfe_{(\bfs_i)}  = (2i)^k \sum_{\bfs_1 \in \mathcal{S}} \cdots \sum_{\bfs_k \in \mathcal{S}}  \prod_{j=1}^k \qty(a_{\bfs_j}  i^{\ordersign{\bfs_j, \bfs_1 \oplus \cdots \oplus \bfs_{j-1} \oplus \bfr}})   \bfe_{\bfs_1 \oplus \cdots \oplus \bfs_k \oplus \bfr},
	\end{align*}
and that 
	\begin{align*}
		f(V)\bfe_{(\bfr)} = \sum_{k=0}^\infty  \frac{(2i)^k}{(k+1)!} \sum_{\bfs_1 \in \mathcal{S}} \cdots \sum_{\bfs_k \in \mathcal{S}}  \qty\Big(\prod_{j=1}^k  a_{\bfs_j}  i^{\ordersign{\bfs_j, \bfs_1 \oplus \cdots \oplus \bfs_{j-1} \oplus \bfr}})  e_{\bfs_1 \oplus \cdots \oplus \bfs_k \oplus \bfr}.
	\end{align*}
Therefore,  the gradient of the loss is given by $\nabla\Loss = \overrightarrow{D}f(V)$.

This representation is the key to the efficient computation of  the gradient from the vector representation of the $D_\bft$.  One first needs to estimate all $D_{\bfs_i}, i\in [p]$  which  can be done with $O(p)$ Hadamard test. Creating the vector $\overrightarrow{D}_\bfs = (D_{\bfs_1}, \cdots, D_{\bfs_p})$  and applying the matrix $f(V)$ give the gradient as summarized in Algorithm \ref{alg:subgroup}. The runtime of the algorithm is $O(p^3+pn)$ because  each column $\bfv_i$ is computed in $O(p n)$ time. Also $f(V)$ can be computed by standard algorithms for matrix functions, such as Schur-based methods or scaling-and-squaring variants, in $O(p^3)$ arithmetic operations.
\end{proof}

\begin{example}
	Consider the ansatz $U(\parameter) = e^{i(a_1 \sigma^1 + a_2 \sigma^2+ a_3 \sigma^3)}$. The set of Pauli strings in the Hamiltonian is $\CS = \{1, 2, 3\}$ which is closed under commutation. Hence, the partial derivatives of the loss are linear combinations of $D_1$, $D_2$, and $D_3$. We can associate normalized $D_1$ with the first canonical vector $e_1 = (1, 0, 0)$ and normalized $D_2$ with $e_2 = (0, 1, 0)$ and normalized $D_3$ with $e_3 = (0, 0, 1)$. Then, the partial derivative is a vector in this space. The gradient is expressed as a vector derived from $\overrightarrow{D} = D_1\bfe_1 + D_2\bfe_2 + D_3\bfe_3$. Figure \ref{fig:matrix illustration} illustrates this example.
\end{example}

\begin{algorithm}[ht]
	\caption{Subgroup Gradient Estimation}\label{alg:subgroup}
	\textbf{Input:} $\CS$
	\begin{algorithmic}[1]
		\Procedure{Subgroup Gradient Estimation}{}
		\STATE Estimate $D_\bfs, \bfs\in \CS$ with Hadamard tests.
		\STATE Compute the matrix $V$ where the column $i$ is computed as in \eqref{eq:V matrix column}.
		\STATE Compute the matrix $f(V),$ where $f(z)=(1-e^z)/z$ (extended continuously at $z=0$).
		\EndProcedure
		\State \textbf{Return} $\widehat{\nabla\Loss} = \overrightarrow{\hat{D}}f(V)$.
	\end{algorithmic}
\end{algorithm}

A simple example of a Hamiltonian closed under commutation is  $A = a_1X^{\tensor n} + a_2 Y^{\tensor n}+ a_3 Z^{\tensor n}.$ The gradient in this case can be computed by a $3\times 3$ matrix and 3 Hadamard tests, irrespective of the dimensionality of the quantum system. Another class is $K$-junta Hamiltonians, that acting non-trivially on $k$ qubits. In that case, all the  Pauli strings $\qps, \bfs\in \CS,$ will be of the form $\sigma_k^\bfs\tensor I_{n-k}$, the set $\CS$ will be closed under commutation.%

\begin{remark}
   When the closedness condition is not directly satisfied, one can apply Theorem \ref{thm:subgroup} to the DLA generated by the terms in $A(\parameter)$. When the dimensionality of the DLA is $m$, Algorithm \ref{alg:subgroup}  computes $\nabla \Loss$  in $O(m^3+dm)$ time with $O(m)$ use of the Hadamard tests.  Hence, Algorithm \ref{alg:subgroup} is efficient when the DLA has dimensionality polynomial in $n$.
\end{remark}

\begin{example}\label{ex:Ising}
	The following Hamiltonian has a DLA with dimensionality $O(n^2)$:
	\begin{align*}
		H = \sum_{i=1}^n Z_iZ_{i+1} + X_i,
	\end{align*}
	where $X_i, Z_i$ are the corresponding Pauli operators. In this case,  applying Theorem \ref{thm:subgroup} to the DLA generated by $H$, Algorithm \ref{alg:subgroup}  computes $\nabla \Loss$  in $O(n^6)$ time with $O(n^2)$ use of the Hadamard tests.
\end{example}

\section{Efficient Estimation with Shadow Tomography}\label{sec:shadow}

Turning the gradient estimation to a series of Hadamard tests has another benefit that can further reduce the number of shots to $O(\log p)$. This can be done using shadow tomography \cite{Aaronson2018,Huang2020,Huang2021,Chen2024,King2025}. The components of the gradient $\nabla\Loss$ correspond to $p$ observables that only depend on $\rho^{out}$ without any reconfigurations. In that case, having several copies of $\rho^{out}$, we can efficiently estimate all the components of the gradient. Based on Theorem \ref{thm:subgroup}, the observables are $p$ different Hadamard tests denoted by $H_{\bfs_j}, j\in [p]$ for measuring $D_{\bfs_j}$.  The following lemma gives an explicit characterization of these observables.%

\begin{lemma}\label{lema:Ds in Hs}
	$D_\bfs = \tr{H_\bfs \rho^{out}}$, where $\rho^{out}$ is the ansatz output state and 
\begin{align}\label{eq:Hs}
	H_{\bfs} = R_{\bfs,-}^\dagger O R_{\bfs,-} - R_{\bfs,+}^\dagger O R_{\bfs,+},
\end{align}
where $O$ is the ansatz observable and $R_{\bfs,\pm} = e^{-i\frac{\pi}{4}}$. 
\end{lemma}
\begin{proof}
    Recall that
	\begin{align*}
		D_\bfs = i\tr{O[\qps, \rho^{out}]}.
	\end{align*}
	Moreover, note from \cite{Mitarai2018} the following property of the commutator for any operator $B$:
	\begin{align*}
		[\qps, B] = i\qty(R_\bfs(\frac{\pi}{2}) B R_\bfs(\frac{\pi}{2})^\dagger - R_\bfs(-\frac{\pi}{2}) B R_\bfs(-\frac{\pi}{2})^\dagger).
	\end{align*}
	With this observation, we obtain
	\begin{align*}
		D_\bfs & = -\tr{O \qty(R_\bfs(\frac{\pi}{2}) \rho^{out} R_\bfs(\frac{\pi}{2})^\dagger - R_\bfs(-\frac{\pi}{2}) \rho^{out} R_\bfs(-\frac{\pi}{2})^\dagger)}   \\
		       & =-\tr{R_\bfs(\frac{\pi}{2})^\dagger O R_\bfs(\frac{\pi}{2}) \rho^{out}} + \tr{R_\bfs(-\frac{\pi}{2})^\dagger O R_\bfs(-\frac{\pi}{2}) \rho^{out} },
	\end{align*}
	where we used the cyclic property of the trace. The last equation gives the expression for $H_{\bfs}$. 
\end{proof}
In order to estimate the gradient, one needs to estimate all the expectation values ${o}^\pm_\bfs:=\tr{R_{\bfs,\pm}^\dagger O R_{\bfs,\pm} \rho^{out}}$ for $\bfs\in \CS$. %
We use CST to exponentially reduce the measurement shots for the gradient estimation when  $O$ has a bounded Hilbert-Schmidt norm. Typically, for non-local bounded norm observables, one can use CST with random Clifford measurements. This approach has  sample complexity $O(\log p)$ but generally when the observables are not classically simulatable it has an exponential computational complexity of $2^{\Theta(n)}$. This is because it performs operations (such as product and trace) on $2^n$ by $2^n$ matrices. Particularly, from Theorem 1 of \cite{Huang2020} we have the following result for estimating the gradient.

\begin{theorem}
	Given an observable $O$, state $\rho^{out}$, and a set $\CS\subseteq \{0,1,2,3\}^n$, the expectation values ${o}^\pm_\bfs:=\tr{R_{\bfs,\pm}^\dagger O R_{\bfs,\pm} \rho^{out}}$ for $\bfs\in \CS$  can be estimated with $\epsilon$ additive error using  $${N  = O\qty( \tfrac{1}{\epsilon^2} \log |\CS|  \max_{\bfs\in \CS} \norm{R_{\bfs,\pm}^\dagger O R_{\bfs,\pm}}_{shadow}^2)}$$ copies of $\rho^{out}$, where $\norm{\cdot}_{shadow}$ is the {\emph shadow norm}.
\end{theorem}
The shadow norm (Definition \ref{def:shadow norm} in Appendix \ref{sec:CST}) is closely related to the variance of the observable and the set of the unitary transformation used for taking the classical shadows. For random Clifford measurements, it is bounded by  the Hilbert-Schmidt  norm; whereas for random Pauli measurements, it is bounded by $4^k$, where $k$ is the  locality of the observable, not the actual number of qubits \cite{Huang2020}. %

\begin{theorem}\label{thm:shadow cliford}
	Consider the ansatz $U(\overrightarrow{a})=\exp{i A(\parameter)}$, with $A(\parameter)=\sum_{\bfs\in {\CS}} \qas \qps$, where $\CS$ is closed under commutation. Suppose $\tr{O^2}\leq B$ for a finite constant $B$. Then,  $\nabla \Loss(\parameter)$ can be estimated with an additive error less than $\epsilon$  using $O(\tfrac{Be^{pc}}{pc\epsilon^2}  \log \frac{|\CS|}{\delta})$ copies of $\rho^{out}$. 
\end{theorem}
\begin{proof}
	We use CST with random Clifford measurements to estimate the observables $H_{\bfs}$ for all $\bfs\in \CA$. From Theorem 1 of \cite{Huang2020},  the estimation of each $D_{\bfs_j}$ with an additive error $\epsilon$ and probability at least $(1-\delta)$ requires $O\qty( \tfrac{1}{\epsilon^2} \log \frac{1}{\delta}  \max_{\bfs\in \CS} \norm{H_{\bfs}}_{shadow}^2)$ copies of $\rho^{out}$. Therefore, 
	it remains to analyze $\norm{H_{\bfs}}_{shadow}^2$. For random Clifford measurements, the shadow norm is bounded by the Hilbert-Schmidt norm \cite{Huang2020}, and hence $\norm{H_{\bfs}}_{shadow}^2 \leq \tr{H_{\bfs}^2}$. We proceed with the following lemma to bound  $\tr{H_{\bfs}^2}$ with the proof given  in Appendix \ref{subsec:proof:lem:tr Hs}.
	\begin{lemma}\label{lem:tr Hs}
		 $\tr{H_\bfs^2}\leq 4\tr{O^2}$ for any $\bfs\in \set{0,1,2,3}^n$. 
	\end{lemma}
	The lemma with the theorem's assumption imply that $\norm{H_{\bfs}}_{shadow}^2 \leq 4B$. Hence, the estimation of all $D_{\bfs}, \bfs\in \CS$ with an additive error $\epsilon$ and probability at least $(1-\delta)$ requires $O\qty( \tfrac{B}{\epsilon^2} \log \frac{|\CS|}{\delta})$ copies of $\rho^{out}$.

	Theorem \ref{thm:subgroup} connects the estimation of all $D_{\bfs}$'s to the estimation of $\nabla \Loss$ via the relation $\nabla\Loss = \overrightarrow{D}f(V)$. Hence, it remains to analyze the error propagation from the estimation of $D_{\bfs}$'s to the estimation of $\nabla \Loss$. We proceed with the following lemma to bound the operator norm of $f(V)$ with the proof given in Appendix \ref{subsec:proof:lem:Bnorm-subgroup}.

	\begin{lemma}\label{lem:Bnorm-subgroup}
Let $V\in \CC^{p\times p}$ be the matrix from Theorem \ref{thm:subgroup}, and  $f(z) := (1-e^{-z})/z$ for $z\neq 0$ and $f(0) := 1$. Then,
\[
\|f(V)\| \le \frac{e^{\|V\|}-1}{\|V\|}.
\]
Moreover, if $|a_{\bfs}|\le c$ for all $\bfs\in\CS$, then
\[
\|f(V)\| \le \frac{e^{2pc}-1}{2pc}.
\]
\end{lemma}
With the lemma, we can bound the error propagation from the estimation of $D_{\bfs}$'s to the estimation of $\nabla \Loss$. Note that $\nabla \Loss = \overrightarrow{D}f(V)$ and hence the estimation error of $\nabla \Loss$ is at most $\|f(V)\|$ times the estimation error of $\overrightarrow{D}$. Hence,  $\nabla \Loss$ can be estimated with an additive error $\epsilon$ and probability at least $(1-\delta)$ using $O\qty( \tfrac{Be^{pc}}{pc\epsilon^2} \log \frac{|\CS|}{\delta})$ copies of $\rho^{out}$.
\end{proof}

\section{Proof of Theorem \ref{thm:loss derivative non product}}\label{sec:proof thm derivative}
We start by taking the partial derivative of $\Loss(\parameter)$. Noting that $\rho^{out}:=U\rho U^\dagger$, the partial derivative of the loss with respect to $\qas$ is 
\begin{align*}
\frac{\partial \Loss}{\partial \qas} = \tr{ O \frac{\partial }{\partial \qas}\rho^{out}},
\end{align*}
where 
\begin{align}\label{eq:rhou out derivative}
\pdv{\rho^{out}}{\qas} = \frac{\partial U}{\partial \qas} \qty(\rho U^\dagger) + \qty(U\rho) \pdv{U^\dagger}{\qas},
\end{align}
and we used the fact that $\rho^{out}$ is Fr\' echet differentiable with respect to $\qas$. To characterize the partial derivative of $U$ we need to introduce some ingredients in Lie algebra and the connection to the derivative of matrix exponential.

We equip the space $\GL(2^n, \CC)$ of  $2^n\times 2^n$ complex matrices with the Lie algebra and the standard Lie bracket. %

Next, we present the following lemma that characterizes the partial derivative of $U$ in terms of the adjoint operator.
\begin{lemma}\label{lem:pdv adj reverse}
Suppose $U=\exp{iH(\tau)}$ for some Hermitian operator valued function $H$. Then, 
\begin{align*}
\dv{U}{\tau} =  -\frac{1-\exp{i\ad_H}}{\ad_H}(\dv{H(\tau)}{\tau}) U,
\end{align*}
where it is assumed that $H(\tau)$ is differentiable.
\end{lemma}
\begin{proof}
    We can write $U=(e^{-iH(\tau)})^\dagger$.  Therefore, given that $\dv{X^\dagger}{\tau}=(\dv{X}{\tau})^\dagger$, from \eqref{eq:exp differential} we can write
    \begin{align*}
    \dv{U}{\tau} &= \qty(\dv{e^{-iH(\tau)}}{\tau})^\dagger = \qty( \exp{-iH(\tau)} \frac{1-\exp{-\ad_{-iH}}}{\ad_{-iH}}\dv{(-iH(\tau))}{\tau})^\dagger\\
    &= \qty(  \frac{1-\exp{-\ad_{-iH}}}{\ad_{-iH}}\dv{(-iH(\tau))}{\tau})^\dagger U.
    \end{align*}
    
    Note that for any $H\in \g$ we have the following equality by its convergent power series:
     \begin{align}\label{eq:exp adjoint}
    \frac{1-\exp{-\ad_H}}{\ad_H} = \sum_{k=0}^\infty \frac{(-1)^k}{(k+1)!} (\ad_H)^k.
     \end{align}
    Therefore, replacing $H$ with $-iH$ in this equation, the derivative equals to the following 
    \begin{align*}
     \dv{U}{\tau} = \qty(-i\sum_{k=0}^\infty \frac{(i)^k}{(k+1)!} (\ad_H)^k (\dv{H}{\tau}))^\dagger  U =i\sum_{k=0}^\infty \frac{(-i)^k}{(k+1)!} \qty((\ad_H)^k (\dv{H}{\tau}))^\dagger  U.
    \end{align*}
    Note that for any $X,Y\in \g$, 
    $(\ad_X(Y))^\dagger = - \ad_{X^\dagger}(Y^\dagger).$
    Therefore,
    \begin{align*}
    \dv{U}{\tau} = i\sum_{k=0}^\infty \frac{(i)^k}{(k+1)!} (\ad_H)^k (\dv{H}{\tau})  U = -\frac{1-\exp{+i\ad_H}}{\ad_H}(\dv{H}{\tau}) U.
    \end{align*}
    \end{proof}
From this lemma, we have that 
\begin{align*}
\frac{\partial U}{\partial \qas} = -\frac{1-\exp{i\ad_{A}}}{\ad_{A}}(\pdv{A}{\qas}) U.
\end{align*}
This gives the first part of \eqref{eq:rhou out derivative}.
Next, from \eqref{eq:exp differential}, the partial derivative of $U^\dagger$ can be  written as
\begin{align*}
\frac{\partial U^\dagger}{\partial \qas} = U^\dagger  \frac{1-\exp{+i\ad_{A}}}{\ad_{A}}(\pdv{A}{\qas}),
\end{align*}
where we used the fact that $A$ is Hermitian. Therefore, the partial derivative of $\rho^{out}$ equals to 
\begin{align*}
\pdv{\rho^{out}}{\qas}  &= -\frac{1-\exp{i\ad_{A}}}{\ad_{A}}(\pdv{A}{\qas}) U \rho U^\dagger + U\rho U^\dagger \frac{1-\exp{i\ad_{A}}}{\ad_{A}}(\pdv{A}{\qas})\\
&=- \frac{1-\exp{i\ad_{A}}}{\ad_{A}}(\pdv{A}{\qas})\rho^{out}  + \rho^{out} \frac{1-\exp{i\ad_{A}}}{\ad_{A}}(\pdv{A}{\qas}).
\end{align*}
By simplifying the terms in the right-hand side, the partial derivative of the loss  can be written as the commutator:
\begin{align}\label{eq:pdv L adjoint}
\frac{\partial \Loss}{\partial \qas} = \tr{ O \frac{\partial }{\partial \qas}\rho^{out}} = \tr{ O \qty[\rho^{out}, \frac{1-\exp{i\ad_{A}}}{\ad_{A}}(\pdv{A}{\qas})]}.
\end{align}
Next, based on the Taylor expansion of the matrix exponential  and from the fact that $\pdv{A}{\qas}=\qps$, the above quantity decomposes as
	\begin{align}\label{eq:pdv adjoint taylor}
		\frac{\partial \Loss}{\partial \qas} = -i\sum_{k=0}^\infty \frac{(i)^k}{(k+1)!}   \tr{ O \qty[\rho^{out}, (\ad_A)^k(\qps)]}.
	\end{align}

Then, building on the binary encoding of Pauli operators, and  Lemma \ref{lem:bracket and K_4}, we can write the adjoint operator as below.
	\begin{lemma}\label{lem:ad decompos}
For any $A,B\in \g$, the corresponding adjoint decomposes as  
\begin{align*}
\ad_A(B) = 2\sum_{\bfr, \bfs} \qas b_\bfr i^{\ordersign{\bfs, \bfr}}  \sigma^{\bfs \oplus \bfr},
\end{align*}
where 
\begin{align*}
\qas:=\frac{1}{2^n}\tr{A\qps} \qquad b_\bfr:=\frac{1}{2^n}\tr{B\sigma^\bfr}
\end{align*}
 are the Pauli coefficients of $A$ and $B$, respectively. 
\end{lemma}
\begin{proof}
  The proof follows from  Lemma \ref{lem:bracket and K_4} and the linearity of adjoint:
  \begin{align*}
    \ad_A(B) = \sum_{\bfs, \bfr} \qas b_\bfr [\qps, \qpr] = \sum_{\bfs, \bfr} \qas b_\bfr 2 i^{\sum_j \ordersign{s_j,r_j}} \sigma^{\bfs \oplus \bfr}. 
  \end{align*}
\end{proof}

Hence, from Lemma \ref{lem:ad decompos} for $\ad_A(\qpr)$, we can write
\begin{align*}
  \ad_A^k(\qpr) = 2^k\sum_{\bfs_1 \in \mathcal{S}} \cdots \sum_{\bfs_k \in \mathcal{S}}  \qty(\prod_{j=1}^k  a_{\bfs_j}  i^{\ordersign{\bfs_j, \bfs_1 \oplus \cdots \oplus \bfs_{j-1} \oplus \bfr}})  \sigma^{\bfs_1 \oplus \cdots \oplus \bfs_k \oplus \bfr}.
\end{align*}
Therefore, the expression for the partial derivative of $\Loss$ equals
\begin{align*}
\frac{\partial \Loss(\parameter)}{\partial \qas} = -i\sum_{k=0}^\infty  \frac{(2i)^k}{(k+1)!} \sum_{\bfs_1 \in \mathcal{S}} \cdots \sum_{\bfs_k \in \mathcal{S}}  \prod_{j=1}^k \qty( a_{\bfs_j}  i^{\ordersign{\bfs_j, \bfs_1 \oplus \cdots \oplus \bfs_{j-1} \oplus \bfr}})  \tr{ O \qty[\rho^{out}, \sigma^{\bfs_1 \oplus \cdots \oplus \bfs_k \oplus \bfr}]}.
\end{align*}
Lastly, recalling that   $D_\bft:=i \tr{O \big[ \sigma^\bft, {\rho}^{out}\big]},$ the above equation simplifies to 
\begin{align*}
  \frac{\partial \Loss(\parameter)}{\partial \qas} = \sum_{k=0}^\infty  \frac{(2i)^k}{(k+1)!} \sum_{\bfs_1 \in \mathcal{S}} \cdots \sum_{\bfs_k \in \mathcal{S}}  \qty(\prod_{j=1}^k  a_{\bfs_j}  i^{\ordersign{\bfs_j, \bfs_1 \oplus \cdots \oplus \bfs_{j-1} \oplus \bfr}})  D_{\bfs_1 \oplus \cdots \oplus \bfs_k \oplus \bfr}.
  \end{align*}

\section{Extensions and Approximations}\label{sec:Extension}
In this section we present extensions and generalizations of the results from the previous sections.

\subsection{Approximations for general case}\label{sec:approximations}
When $\CS$ in $A(\parameter)$ is not a subgroup, then the terms appearing in the partial derivative in  Theorem \ref{thm:loss derivative non product} expand beyond $\CS$. Indeed, they form a subgroup denoted by $\<\CS\>$ called the subgroup generated by $\CS$ and is defined as the smallest subgroup containing $\CS$. In that case, one can compute the gradient from Theorem \ref{thm:subgroup} and Algorithm \ref{alg:subgroup} for $\<S\>$. This yields a run-time that scales polynomially with $|\<\CS\>|$. This is tractable as long as the size of $\<\CS\>$ is polynomial in $d$. Otherwise, one needs a different approach. In what follows, we introduce approximations to the gradient computation for general $\CS$.

\subsubsection{Truncation}
We propose to approximate the partial derivatives by truncating the infinite sum that appeared as \eqref{eq:pdv infinite sum} in Theorem \ref{thm:loss derivative non product}.
\begin{theorem}\label{thm:derivative approx}
	Consider an ansatz of the form $U(\overrightarrow{a})=\exp{i A(\parameter)}$, where $A(\parameter) = \sum_{\bfs \in \mathcal{S}} \qas \qps$, and $\mathcal{S}$ is the index of the parameters.   Then, the derivative of $\Loss(\parameter)$ for this ansatz is approximated in terms of  $D_\bft$ as:
	\begin{align*}
		\frac{\partial \Loss}{\partial a_{\bfr}} = \sum_{k=0}^K  \frac{(2i)^k}{(k+1)!} \sum_{\bfs_1 \in \mathcal{S}} \cdots \sum_{\bfs_k \in \mathcal{S}}  \qty\Big(\prod_{j=1}^k  a_{\bfs_j}  i^{\ordersign{\bfs_j, \bfs_1 \oplus \cdots \oplus \bfs_{j-1} \oplus \bfr}})  D_{\bfs_1 \oplus \cdots \oplus \bfs_k \oplus \bfr}, + \epsilon_K,
	\end{align*}
	where $\epsilon_K = \exp{-K\log \frac{K}{ e \pi p \norm{O}}}$.
\end{theorem}

\begin{corollary}
	In the setting of the above theorem, with bounded $\norm{O}$, having $K=\Omega(p+\log \frac{1}{\epsilon})$ suffices to approximate $\pdv{\Loss}{\qas}$ upto an $\epsilon$ additive error.
\end{corollary}

We present a numerical experiment to study the effect of the truncation on the approximation error of the gradient.  For that we generated a Hamiltonian
\begin{align*}
	H(\parameter) = \sum_{i=1}^d a_i Z_iZ_{i+1} + a^\prime_{i}X_i,
\end{align*}
where $Z_i$ and $X_i$ are the Pauli operators applied on the $i$'th qubit. Then we use our method with various truncation parameters $K$ to calculate the estimation error. Fig. \ref{fig:approx} demonstrates the averaged approximation error measured as the norm distance between the true gradient and our approximation. The graphs are for randomly chosen parameters $a_i, a^\prime_i$  as a function of $K$ for various qubit numbers $n$. The figure shows that with a few number of parameters, one can approximate the gradient with high precision. Shaded regions indicate one standard deviation over 10 trials with random observables $O$.

\begin{figure}[ht]
	\centering
	\includegraphics[width =0.5 \textwidth ]{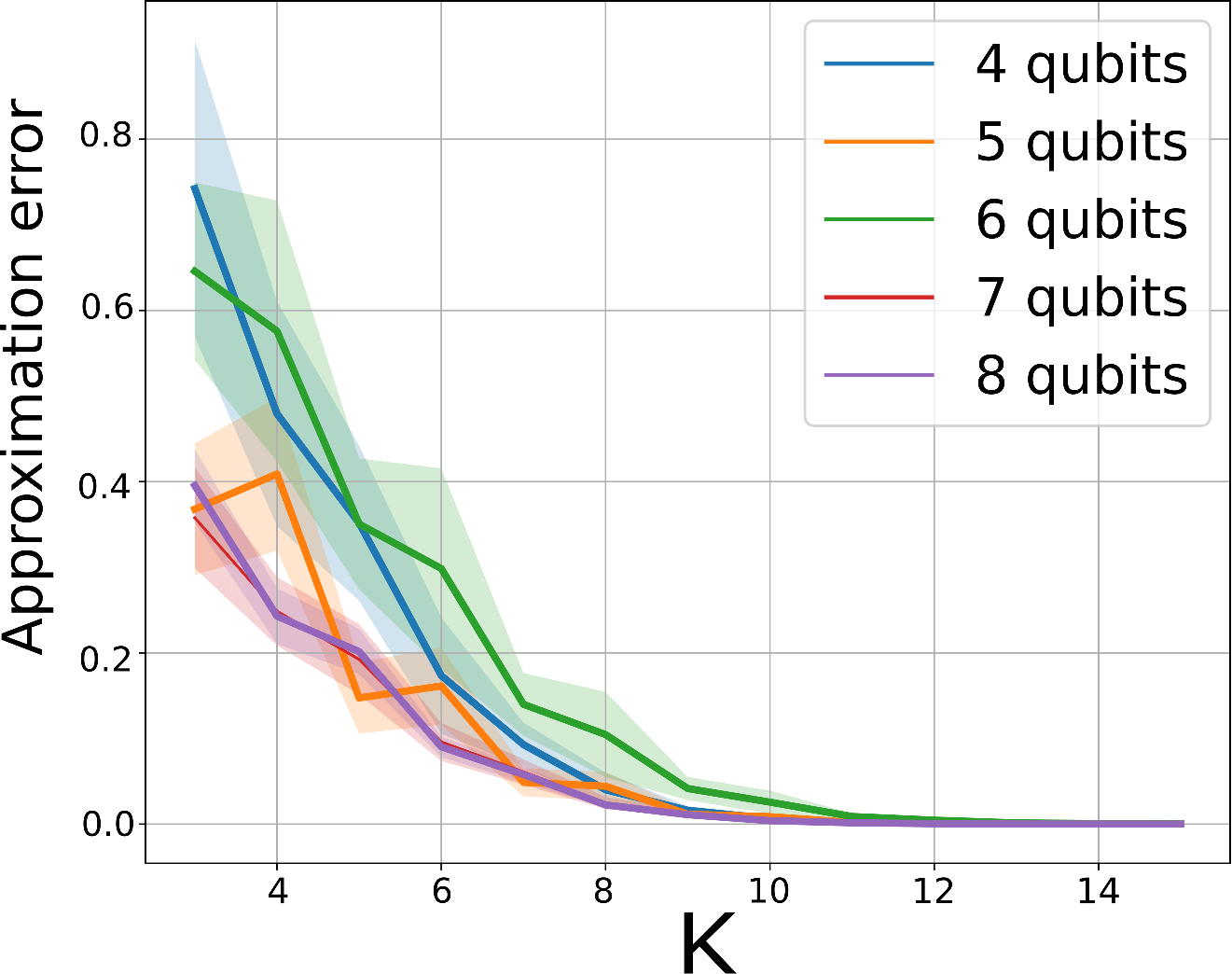}
	\caption{Estimation error for gradient approximation as a function of $K$ for various Hamiltonians with different numbers of qubits. }
	\label{fig:approx}
\end{figure}

\subsubsection{Unbiased Estimation via Randomization}
The infinite sum can be estimated via a randomization technique. We show that the partial derivative in Theorem \ref{thm:loss derivative non product} can be written as an expectation value of a function of the Poisson random variable. Let $K$ be the Poisson random variable  with rate $\lambda=2$:
\begin{align*}
	\prob{K=k}= \frac{2^k}{k!}e^{-2}, \qquad k=0,1,\cdots
\end{align*}
By sampling $K$ from this distribution we can estimate the partial derivative expression of Theorem \ref{thm:loss derivative non product}. Note that
\begin{align}\label{eq:pdv infinite sum randomized}
	\frac{\partial \Loss}{\partial a_{\bfr}} = \sum_{k=0}^\infty \frac{(2)^{k+1}}{(k+1)!} X(k)
\end{align}
where
\begin{align*}
	X(k):=\frac{(i)^k}{2} \sum_{\bfs_1 \in \mathcal{S}} \cdots \sum_{\bfs_k \in \mathcal{S}}  \prod_{j=1}^k \qty( a_{\bfs_j}  i^{\ordersign{\bfs_j, \bfs_1 \oplus \cdots \oplus \bfs_{j-1} \oplus \bfr}})    D_{\bfs_1 \circ \cdots \circ \bfs_k \circ \bfr}.
\end{align*}
In that case, $$\frac{\partial \Loss}{\partial a_{\bfr}} =e^2\EE[X(K)].$$ Hence, by sampling from this distribution we can compute an unbiased estimate of the partial derivative.

\subsection{Short-term Hamiltonian}
The above procedure can be made significantly more efficient if one guarantees that the parameter range is small, that is $|\qas|\leq \epsilon$ for small enough $\epsilon>0$. This is the case in short-term Hamiltonian simulation or a Hamiltonian with ``weak interactions". In the short-term simulation, the unitary is of the form $U=e^{i\Delta t H}$, where $\Delta t$ is the time step. In that case, each partial derivative is approximated in $O(1)$ classical computation with one Hadamard test of Figure \ref{fig:loss derivative circuit multi}.
\begin{corollary}
	If $\max_{\bfs\in \CS}|\qas|\leq \frac{\epsilon}{p (1+2\epsilon)}$, then the partial derivative can be approximated with $\epsilon$ additive error via $\pdv{\Loss}{\qas} \approx D_{\bfs}$.
\end{corollary}

\subsection{Extension to Polynomial Size DLA}\label{sec:DLA}
Next, we extend our results to a more  general class of PQCs for which DLA has $\poly(d)$ dimensionality.
\begin{definition}\label{def:generators}
	For  the unitary of the form $e^{iA(\parameter)}$ where $A(\parameter) = \sum_i a_i G_i$ with $G_i$ being traceless Hermitian operators, the set $\CG = \{G_1, \cdots, G_p\}$ is called the generators of the PQC.
\end{definition}

\begin{definition}[Dynamical Lie Algebra]
	Given a PQC with generators $\CG$ as in Definition \ref{def:generators}, the Dynamical Lie Algebra (DLA) $\g_A$ is the subalgebra of
	$\su(2^n)$ spanned by the repeated nested commutators of the elements in $\CG$, i.e.,
	\begin{align*}
		i\g_A := \emph{Span}_\RR\<iG_1, \cdots, iG_p\>_{\text{Lie}} \subseteq \su(2^n),
	\end{align*}
	where $\<\cdot\>_{\text{Lie}}$ denotes the Lie closure, obtained by repeatedly taking the nested commutators, and the span is over the real numbers. We say that $\g_A$ has polynomial size/dimensionality if $\dim(\g_A) = \poly(n)$ as a vector space.
\end{definition}

In general, computing the DLA is computationally expensive. One viable approach involves direct construction starting from a set of generators and iteratively commuting them to discover new elements until a basis of the DLA is attained (for example see Algorithm 1 in \cite{Larocca2022}). The complexity of this approach in general is $O(2^n)$. However, this is not an issue if the dimension of the DLA to be constructed is promised to by $\poly(n)$.

\begin{theorem}\label{thm:DLA derivative}
	Consider an ansatz of the form $U(\overrightarrow{a})=\exp{i A(\parameter)}$, where $A(\parameter)=\sum_{i} a_i G_i$ with $G_i$ being traceless Hermitian operators. Let $\g_A$ be the DLA of the ansatz with a basis $\{E_1, \cdots, E_{d_{\g_A}}\}$ and $T_A$ be the matrix representation of the linear transformation $\ad_A(\cdot) = [A, \cdot]$ in the basis of $\g_A$.
	Then, $\nabla \Loss(\parameter)=\overrightarrow{R}f(T_A)$, where $f(z) := (1-e^{-z})/z$ for $z\neq 0$ and $f(0) := 1$, and $\overrightarrow{R} = (R_1, \cdots, R_{d_{\g_A}})$ with $R_l = i\tr{O[\rho^{out}, E_l]}$
	
	Moreover, if $\g_A$ has $\poly(n)$ dimensionality, there is an algorithm (Algorithm \ref{alg:DLA}) that estimates $\nabla \Loss(\parameter)$ with $\poly(n)$ copies of $\rho^{out}$ and additional $\poly(n)$ classical time. 

\end{theorem}
\begin{proof}
	The proof of the theorem follows from a similar argument as in Theorem \ref{thm:loss derivative non product} and \ref{thm:subgroup}. First note that from \eqref{eq:pdv L adjoint}
	\begin{align*}
		\frac{\partial \Loss}{\partial a_j} =  i\tr{ O \qty[\rho^{out}, \frac{1-\exp{i\ad_{A}}}{\ad_{A}}(\pdv{A}{a_j})]}.
	\end{align*}
	From the definition of $\frac{1-e^{i\ad_A}}{\ad_A}$ and the fact that $\pdv{A}{a_j}=G_j$, the above quantity decomposes as
	\begin{align*}
		\frac{\partial \Loss}{\partial a_j} = -i\sum_{k=0}^\infty \frac{(i)^k}{(k+1)!}   \tr{ O \qty[\rho^{out}, (\ad_A)^k(G_j)]},
	\end{align*}
	Note that $\ad_A$ maps $\g_A$ unto itself as its \textit{commutator ideal}. That is $$\ad_A(G) = \sum_{j=1}^{m} c_j [V_j, W_j],$$ where $V_j, W_j\in \g_A$. Since $\g_A$ is closed under the Lie bracket, then  $[V_j, W_j]\in \g_A$. Hence, $\ad_A(G)\in \g_A$ and it equals to a finite sum of the basis elements $E_l$ of $\g_A$ as
	\begin{align*}
		\ad_A(G_j)=\sum_{l=1}^{d_\g} \alpha_{l,j} E_l.
	\end{align*}
	Therefore, the above summation can be rewritten as
	\begin{align*}
		\frac{\partial \Loss}{\partial a_j} = -i\sum_{k=0}^\infty \frac{(i)^k}{(k+1)!}  \sum_{l=1}^{d_{\g_A}} \alpha_{l,j} \tr{ O \qty[\rho^{out}, E_l]}.
	\end{align*}

	Note that $\ad_A(G_j)$ can be viewed as a vector $(\alpha_{1,j},\cdots, \alpha_{d_{\g_A},j})$ in the $E_l$ basis. Moreover, $\ad_A$ can be viewed as a linear transformation (matrix) $T_A(\overrightarrow{\alpha})$ in the $E_l$ basis. This matrix is $d_{\g_A}\times d_{\g_A}$. With that taken into account, the partial derivative can be written as
	\begin{align*}
		\pdv{\Loss}{a_{j}} \equiv \sum_{k\geq 0} \frac{(-i)^k}{(k+1)!} T_A^k(\overrightarrow{\alpha}) G_{j}, \qquad \forall j\in [p].
	\end{align*}
	Where $G_j$ is understood as a vector in the $E_l$ basis. Now let $$B(\overrightarrow{\alpha}) = f(T_A(\overrightarrow{\alpha})),$$
	where $f(z) := (1-e^{-z})/z$ for $z\neq 0$ and $f(0) := 1$. This matrix can be computed classically in $\poly(d_{\g_A})=\poly(n)$ time. Note that the above summation is exactly the Taylor expansion of $f(T_A(\overrightarrow{\alpha}))$, and hence we have
	\begin{align*}
		\pdv{\Loss}{a_{j}} = \sum_{l=1}^{d_{\g_A}}  B_{(j,l)}(\overrightarrow{\alpha})  i\tr{O[E_l, \rho^{out}]},
	\end{align*}
	where $B_{(j,l)}(\overrightarrow{\alpha})$ is the $(j,l)$ entry of the matrix $B(\overrightarrow{\alpha})$.
	  To compute the gradient from this vector representation, one first needs to estimate all $R_{l}:=i\tr{ O \qty[\rho^{out}, E_l]}$ and compute the following
	\begin{align*}
		\nabla\Loss =\overrightarrow{R}  B(\overrightarrow{\alpha}) .
	\end{align*}
	With that the proof is complete.
\end{proof}

The following algorithm summarizes this procedure to estimate the gradient.
\begin{algorithm}[ht]
	\caption{DLA Gradient Estimation}\label{alg:DLA}
	\textbf{Input:} DLA generators $E_1, \cdots, E_{d_{\g_A}}$,

	\begin{algorithmic}[1]
		\STATE Estimate the expectation value of the tests $R_l:= i\tr{O[E_l, \rho^{out}]}$ for $l\in [d_{\g_A}]$.
		\STATE Compute the $d_\g\times d_{\g_A}$ matrix $B(\overrightarrow{\alpha})=f(T_A(\overrightarrow{\alpha}))$, with $f(z) := (1-e^{-z})/z$ for $z\neq 0$ and $f(0) := 1$.
		\State \textbf{Return} $\widehat{\nabla\Loss} =  \overrightarrow{R}B(\overrightarrow{\alpha})$.
	\end{algorithmic}
\end{algorithm}

Extension to PQC with mixture of different DLA is also possible. For example, consider a PQC of the form $U(\parameter) = \exp{i A_1(\parameter)}\exp{i A_2(\parameter)}$ where $A_1(\parameter) = \sum_i a_i G_i$ and $A_2(\parameter) = \sum_j a^\prime_j G^\prime_j$ with $\{G_i\}$ and $\{G^\prime_j\}$ being the generators of two different DLA's $\g_{A_1}$ and $\g_{A_2}$. If both $\g_{A_1}$ and $\g_{A_2}$ have polynomial size/dimensionality, then one can compute the gradient with polynomial number of measurement tests and additional polynomial classical time. The proof follows from a similar argument as in Theorem \ref{thm:DLA derivative}.

\section{Concluding Remarks}
This paper provides a framework to estimate the gradient of generic PQCs via Hadamard tests for Pauli operators followed by classical post-processing. It is shown that the proposed approach is polynomial in classical and quantum resources when the DLA of the associated Hamiltonian of the PQC has a dimensionality polynomial in the number of qubits. Moreover, this method does not change the ansatz structure and can be used to reduce the measurement shot complexity to scale logarithmically with the number of parameters. The results would be beneficial in various optimization or learning quantum algorithms that rely on the estimation of the gradient.

Our results are steps toward a unified framework for efficiently estimating the gradient of an arbitrary parameterized circuit without making changes to its unitary. 
One limitation of this work is when the Hamiltonian has exponentially many Pauli terms. In that case, we can only approximate the gradient by truncation of the nested summations in Theorem \ref{thm:loss derivative non product} to a fixed number of terms. However, this will be a biased approximation.

As future work, one can extend the proposed framework to the estimation of higher-order derivatives such as the Hessian or the Fubini-Study metric defined for a parameterized quantum system. For that, one needs to define a second-order Hadamard test strategy for measuring the double derivatives. Another future work is to extend this strategy to multi-layered PQCs where each layer might be a black box unitary. Finally, deriving lower bounds on the classical and quantum resources needed to estimate the gradient or higher-order derivatives is another important direction.

\section*{Acknowledgments}
This work is partially supported by the NSF Center for Science of Information (CSoI) Grant
CCF-0939370, and also by NSF Grants CCF-2006440 and and CCF-2211423.

\bibliographystyle{quantum}
\bibliography{main.bbl}
\newpage
\appendix
\section{Derivation of The Gradient for a Single Qubit PQC}\label{sec:example single qubit}
	Consider a general single-qubit unitary of the form
	\begin{equation*}
		U(\overrightarrow{a}) = \exp{i(a_1 \sigma^1 + a_2 \sigma^2 +a_3\sigma^3)}.
	\end{equation*}
	Let $O$ be a generic observable and consider the associated loss $\Loss(\parameter)$ as in \eqref{eq:VQA cost function}. As an illustrative  example, we examine the partial derivative of $\Loss$ with respect to $a_1$, evaluated  at the point $a_1=0$ and $a_3=0$. In the following we compute this derivative using two approaches: first, by applying Theorem \ref{thm:loss derivative non product}, and second,through direct calculation.
    
    \paragraph{Closed-form expression based on Theorem \ref{thm:loss derivative non product}.}
    In the context of Theorem \ref{thm:loss derivative non product}, let $D_j$ be the result of the Hadamard test with Pauli $\sigma^j$, where $j=1,2,3$. Then, \eqref{eq:pdv infinite sum}  in Theorem \ref{thm:loss derivative non product} simplifies to the following:
	\begin{align*}
		\frac{\partial \Loss}{\partial a_1}(\parameter = (0,a_2,0)) = \sum_{k=0}^\infty \frac{(2i)^k}{(k+1)!}  \qty(\prod_{j=1}^k a_{2}  i^{\ordersign{2, (\overbrace{2 \oplus \cdots \oplus 2}^{j-1 ~\text{times}}\oplus 1)}})   D_{(\underbrace{2 \oplus \cdots \oplus 2}_{k ~\text{times}}\oplus 1)},
	\end{align*}
	where we used the fact that only terms with $\bfs_j = 2$ are surviving. Because $a_1=a_3=0$.
	Note that for even $j$ we have $\underbrace{2 \oplus \cdots \oplus 2}_{j-1 ~\text{times}} \oplus 1 = 3$, and for odd $j$ it is equal to $0\oplus 1 = 1$. Therefore,
	\begin{align*}
		\ordersign{2, (\underbrace{2 \oplus \cdots \oplus 2}_{j-1 ~\text{times}}\oplus 1)} = \begin{cases}   \ordersign{2, 3} = 1 & \text{even}~ k \\
			\ordersign{2, 1}  =-1 & \text{odd}~ k
		                                                                                 \end{cases}
	\end{align*}
	where we used \eqref{eq:ordersign}. Plugging it in the first equation, we have

	\begin{align*}
		\prod_{j=1}^k a_{2} \qty( i^{\ordersign{2, (\overbrace{2 \oplus \cdots \oplus 2}^{k-1 ~\text{times}}\oplus 1)}})  = a_2^k (-i)\times i \times (-i) \times \cdots
		 & = \begin{cases} (a_2i)^k  (-1)^{k/2}     & \text{even}~ k \\
              (a_2i)^k  (-1)^{(k+1)/2} & \text{odd}~ k.
		     \end{cases}
	\end{align*}

As a result, the partial derivative simplifies to 
	\begin{align*}
		\frac{\partial \Loss}{\partial a_1}(\parameter = (0,a_2,0)) =
		\sum_{p=0}^\infty \frac{(-2)^{2p}}{(2p+1)!} a^{2p}_{2} (-1)^p  D_{1} + \sum_{q=0}^\infty \frac{(-2)^{2q+1}}{(2q+2)!} a^{2q+1}_{2} (-1)^{q+1}  D_{3}.
	\end{align*}
	Hence, one needs to measure $D_1$ and $D_3$ to compute the above partial derivative.  Next, by simplifying the summations, it is not difficult to show that
	\begin{align*}
		\frac{\partial \Loss}{\partial a_1}(\parameter = (0,a_2,0)) & =   \frac{-1}{2a_2} \qty( \sum_p \frac{(-2a_2)^{2p+1}}{(2p+1)!} (-1)^p )D_1 + \frac{-1}{2 a_2} \qty( \sum_q \frac{(-2a_2)^{2q+2}}{(2q+2)!} (-1)^{q+1} )D_3 \\
		                                                            & = \frac{-1}{2a_2} \qty\big( \sin(-2a_2) D_1  +  (\cos{(-2a_2)}-1) D_3)                                                                                     \\
		                                                            & = \frac{1}{2a_2} \qty\big( \sin(2a_2) D_1  +  (1- \cos{(2a_2)}) D_3).
	\end{align*}
	Notice the presence of $D_3$ which relates to the Pauli $\sigma^3$ and not appear in the ansatz expression.
	One can verify that this is indeed equal to the analytic gradient of this ansatz.

\paragraph{Direct derivation.}
Note that from \eqref{eq:rhou out derivative} the partial derivative of the objective function can be written as 
\begin{align*}
\frac{\partial \Loss}{\partial \qas} = \tr{ O \qty(\frac{\partial U}{\partial \qas} \qty(\rho U^\dagger) + \qty(U\rho) \pdv{U^\dagger}{\qas})},
\end{align*}
Given that $\pdv{U^\dagger}{\qas}=(\pdv{U}{\qas})^\dagger$. Then, by denoting $\tilde{U}=\pdv{U}{\qas}$ we have that 
\begin{align*}
\pdv{\Loss}{\qas} &= \tr{ O \tilde{U}\rho U^\dagger + U\rho\tilde{U}^\dagger}.
\end{align*}
Next, as $UU^\dagger = I$, we have 
\begin{align*}
 \pdv{\Loss}{\qas} &= \Tr{O\left( \tilde{U} U^\dagger (U \rho U^\dagger) + (U \rho U^\dagger) U \tilde{U}^{\dagger}\right)} \\
    &= \Tr{O\left( \tilde{U} U^\dagger \rho^{out} + \rho^{out} (\tilde{U} U^\dagger)^\dagger\right)},
\end{align*}
where $\rho^{out}$ is the ansatz output. 
Note that the single qubit ansatz  can also be written as
\begin{equation} \label{eq:single qubit unitary}
    U(\overrightarrow{a}) = I\cos{\theta} + i\left( \sum_{s \in \{0,1,2,3\}} \hat{a}_s \sigma^s\right) \sin{\theta},
\end{equation}
where $\theta = \sqrt{\sum a_s^2}$ is a normalizing parameter and $\hat{a}_s = \frac{a_s}{\theta}$. Now, we can differentiate $U$ with respect to a single parameter $a_s$ appearing in the sum:
\begin{equation} \label{eq:deriv single unitary}
\begin{split}
    \frac{\partial U(\overrightarrow{a})}{\partial a_s} = (-\frac{a_s}{\theta}\sin{\theta}) I &+ i\left( \sum_{s^\prime \neq s}\frac{-a_sa_{s^\prime}}{\theta^3}\sigma^{s^\prime}  + \frac{\theta^2 - a_s^2}{\theta^3}\sigma^s\right)\sin{\theta} \\ &+ i \left( \frac{a_s\cos{\theta}}{\theta}\sum_{s^\prime} \hat{a}_{s^\prime} \sigma^{s^\prime}\right).
\end{split}
\end{equation}
Therefore,
\begin{align*}
	\left.\frac{\partial U(\parameter)}{\partial a_1}\right|_{a_1=a_3=0} = i\sigma^1 \frac{\sin a_2}{a_2}.
\end{align*}

Using \eqref{eq:single qubit unitary} to find $U^\dagger$, we have that 
\begin{equation*}
    \tilde{U} U^\dagger = i \frac{\sin{a_2}}{a_2} \left( \cos{a_2} \sigma^1 + \sin{a_2} \sigma^3\right).
\end{equation*}
Which when plugged into the derivative expression gives
\begin{equation*}
    \left.\frac{\partial L}{\partial a_1}\right|_{a_1=0} = \frac{\sin{a_2}}{a_2} \left( \cos{a_2} D_1 + \sin{a_2} D_3 \right)= \frac{\sin{2a_2}}{2a_2} D_1 + \frac{(1-\cos 2a_2)}{2a_2}D_3.
\end{equation*}
This is identical to the expression based on Theorem \ref{thm:loss derivative non product}.

\section{Summary of Related Works}\label{sec:related}

\paragraph{Parameter Shift Rule.}
There have been several approaches to estimating the gradient \cite{Farhi2018,Mitarai2018,Sweke2020,HeidariAAAI2022,Harrow2021,Schuld2018,Mitarai2019,Wiersema2024}. The zeroth-order approach (e.g., finite differences) evaluates the objective function in the neighborhood of the parameters.  Although it is a generic approach, recent studies showed their drawbacks in terms of convergence rate \cite{Harrow2021}.  First-order methods (e.g., parameter shift rule)  directly calculate the partial derivatives \cite{Schuld2018}. When the ansatz is of the form $U(\parameter) = \prod_{j=1}^p e^{-i a_j \sigma^{\bfs_j}}$, the parameter shift rule implies that:
\begin{equation*}
	\pdv{\Loss}{a_j} = \<\Loss(\parameter + \frac{\pi}{4}e_j)\> -\<\Loss(\parameter - \frac{\pi}{4}e_j)\>,
\end{equation*}
where $e_j\in \RR^p$ is the $j$th canonical vector, that is $e_{j,j}=1$ and $e_{j,r}=0$ for all $r\neq j$.

\paragraph{Stochastic \ac{PSR}:}
Stochastic parameter shift rule was proposed in \cite{Banchi2021} to measure the derivative for unitary operators of the form $e^{i (\theta \qps + B)}$, where $\qps$ is a Pauli string and $B$ is an arbitrary Hermitian operator. The authors showed that the derivative is equal to
\begin{align*}
	\pdv{\Loss}{\theta}  = \int_0^1 C_+(\theta, t) - C_-(\theta, t) d t,
\end{align*}
where $C_\pm(\theta, t) : = \tr{O V_\pm(\theta, t) \rho V_\pm^\dagger(\theta, t)}$, with
\begin{align*}
	V_\pm(\theta, t) :=  e^{i t (\theta \qps + B)}e^{\pm i \qps} e^{i(1-t)(\theta \qps + B)}.
\end{align*}
Then the authors present a Monte Carlo strategy to estimate this integral. A more general variant of the parameter shift rule has been introduced \cite{Wierichs2022}. The considered general ansatz of the form $U(\theta) = e^{i (\theta A + B)}$ for generic Hermitian $A$ and $B$. Given the spectral decomposition of $A$ the paper provides an explicit formula for the derivative of the objective function. This is done via a Discrete Fourier series approach.

Such methods for generic ansatz $e^{i (\theta A + B)}$ require not only to perturb the parameters but also to change the unitaries involved. As a result,
to modify not just the parameters of the \ac{PQC}, but also change the unitaries that appear. In practice, such modifications require a re-evaluation of the schedule of the underlying quantum-control system and hence are at a disadvantage. Moreover, the stochastic \ac{PSR} has a high estimation variance. This is because the above integral is estimated by sampling values of $s$ uniformly in the interval $(0, 1)$ and then
calculating the costs with a finite-shot estimate. In addition, this method leads to a bigger number of unique circuits to compute the derivative, increasing the compilation overhead for both hardware and simulator implementations.

\paragraph{Nyquist \ac{PSR}:} Recently  \cite{Theis2023} proposed a ``proper" shift rule for \acp{PQC} of the form $e^{i (\theta A + B)}$ where only the parameters are shifted without any other modifications of the ansatz. The method was called \textit{Nyquist parameter shift rule} and relies on a beautiful connection between the Nyquist-Shannon Sampling theorem and the Fourier series that was observed earlier in \cite{Wierichs2022,Vidal2018}. This paper shows that if $f(x) = tr\{O U \rho U^\dagger \}$ with $U=e^{ix H + B}$ and $K$ being the difference between the maximum and minimum eigenvalues of $A$,  then the Fourier spectrum of $f$ is contained in $[-K, K]$. Hence, the Nyquist-Shannon Sampling theorem can be used to estimate the derivative of the objective function.  The proposed method, however, has a low approximation error when the parameter value is large enough. More precisely, the approximation error is $O(\frac{1}{c^2})$ as long as $\theta = (1 - \Omega(1))c$, where $c$ is the maximum magnitude of a parameter value. Our method is suitable when the parameter value is small.

\paragraph{Trotterization.}
In this approach the non-product unitary is approximated via the Suzuki-Trotter transformation \cite{Suzuki1976} which states that for any operators $A_1, A_2, \cdots, A_k$, that do not necessarily commute with each other, the following holds
\begin{equation*}
	\lim_{n\rightarrow \infty} \Big(\prod_{j=1}^k \exp{A_j/n }\Big)^n = \exp\big\{\sum_{j=1}^k A_j\big\}.
\end{equation*}
The Trotter formula has been used in literature to derive approximations for generic PQCs \cite{Yang2022,Liu2020,Miessen2021}. However, implementing the above approximation may lead to a high gate complexity and is not preferable in scenarios where direct implementation is available.

\section{Classical Shadow Tomography}\label{sec:CST}
For completeness, in this section we briefly describe the classical shadow tomography procedure. For more details see \cite{Huang2020}. Classical shadow tomography is a technique used in quantum computing to efficiently learn properties of a quantum state using only a few measurements. It was introduced to extract useful information from quantum states without requiring a full quantum state tomography, which is costly in terms of the number of measurements and computational resources. 

More precisely, let $O_j,  j\in[M]$ be a set of observables. The goal is to estimate the expectation value of these observable for measuring an unknown state $\rho$ in a Hilbert space $\mathcal{H}$. For that, several copies of $\rho$ are provided. CST is a procedure with minimal sample (copy) complexity.

\begin{theorem}[\cite{Huang2020}]
	Suppose the observables $O_j, j\in M$ are traceless, then the expectation values $\tr{O_j\rho}, j\in [M]$ can be approximated up to an additive error $\epsilon$ with probability $(1-\delta)$ given $$O\qty(\frac{1}{\epsilon^2} \log \frac{M}{\delta} \max_j \norm{O}_{\text{shadow}}^2)$$ copies of $\rho$.
\end{theorem}
The shadow norm is a measure that resembles the variance in the worst case state.  

In what follows we describe the steps in this procedure for a generic state $\rho$ in a Hilbert space $\mathcal{H}$.. 

First, generate a  unitary operator $U$ randomly from a class of choices $\mathcal{U}$ to be determined.  %
Apply $U$ on the input state resulting in the state $U^\dagger \rho U$. Measure the resulted state in the canonical basis $\ket{j}, j \in [\dimH]$. From Born’s rule the probability of getting  the output $j$ is $p_{j}= \matrixelement{j}{U^{\dagger}\rho U}{j}$. 
Given an outcome $j$, define  $\omega_j = U\ketbra{j}U^\dagger.$ %
The expectation   $\EE_{\sim (J, U)}[\omega_J]$  over the measurement randomness ($p_j$) and the choice of  $U$ equals to $\CM[\rho],$ where $\CM$ is a mapping defined as
\begin{align}\label{eq:Gamma}
\CM[O]:=\EE_{U}\Big[\sum_{j\in [\dimH]} \matrixelement{j}{U^{\dagger}O U}{j}~ U\ketbra{j}U^\dagger\Big],
\end{align}
for any operator $O$ on $\mathcal{H}$. 
Observe that $\CM$ is a linear mapping on $\mathcal{B}(\mathcal{H})$ and hence has an inverse denoted by $\CM^{-1}$. We note that $\CM^{-1}$ is the shadow channel $\mathcal{M}^{-1}$ introduced in \cite{Huang2020}. %
We apply $\CM^{-1}$ on $\omega_j$ resulting in the so called shadow
\begin{align}\label{eq:rho hat shadow}
\hat{\rho}:=\CM^{-1}\big[U\ketbra{j}U^\dagger\big].
\end{align}
Note that $\hat{\rho}$ is a classical matrix and hence can be copied several times. Moreover, $\hat{\rho}$ is not a valid density operator as it is not necessarily a positive semi-definite matrix. However, it is an unbiased estimate of the original state. 

When $\mathcal{U}$ is tomographically complete, the classical shadow $\hat{\rho}$ is unbiased, that is $\EE_{U, J}[\hat{\rho}]=\rho$. %

\begin{definition} \label{def:shadow norm}
	The shadow norm of any operator $O$ on $\mathcal{H}$ is %
	\small
	\begin{align*}
	\norm{O}_{\emph{shadow}}: = \max_{\sigma\in \mathcal{D}[H]}\Big(\sum_{j\in [\dimH]}\hspace{-10pt} \matrixelement{j}{U^{\dagger}\sigma U}{j}~ \expval{U\CM^{-1}[O]U^\dagger}{j}^2\Big)^{1/2}. 
	\end{align*}
	\normalsize
	\end{definition}

\subsection{CST with Pauli Measurements}\label{subsec:CST Pauli}
\noindent\textbf{Shadow tomography with Pauli measurements.} Suppose the observable $O_j$ act non trivially on at most $k$ qubits. For that $V$ in CST is the tensor product of randomly chosen Pauli operators:
\begin{align*}
	V = V_1\tensor \cdots \tensor V_d \in CL(2)^{\tensor d},
\end{align*}
where each $V_j$ is chosen randomly and uniformly from the Clifford group $CL(2)$. In this case, $\CU=CL(2)^{\tensor d}$. Moreover,  the shadow matrix is computed as
\begin{align*}
	\hat{\rho}:=\bigotimes_{j=1}^d \qty(3V_j^\dagger \ketbra{\hat{b}_j}V_j - I).
\end{align*}
In that case, the shadow norm is bounded by the locality of the observables as 
\begin{align*}
	\norm{O_j}_{\emph{shadow}}\leq 2^k\norm{O_j}_\infty
\end{align*}
As a result the sample complexity of CST is given by   $$n=O\qty(\frac{4^k}{\epsilon^2}\log m \max_j \norm{O_j}^2_\infty).$$
  The CST algorithm runs in $\tilde{O}\qty(2^{\Theta(k)}m\log m)$ classical time.

\subsection{CST with Clifford Measurements}
The \textit{Clifford group} is a set of unitary operations that map Pauli operators to other Pauli operators under conjugation. For a system of \(d\) qubits, the Clifford group \( \mathcal{C}_d \) consists of unitaries \(U\) such that for any Pauli string \( P \):
\[
U P U^\dagger = P',
\]
where \(P'\) is another Pauli operator.

Clifford circuits are particularly useful because they can be efficiently simulated classically, and their structure allows for easy manipulation of Pauli observables, making them useful for shadow tomography. In that case, $V$ is chosen randomly from the Clifford group. The shadow norm is bounded as 
\begin{align*}
	\norm{O_j}_{\emph{shadow}}\leq \sqrt{3\tr{O^2}},
\end{align*}
when $\tr{O^2}<\infty$. The shadow matrix is given by 
\begin{align*}
	\hat{\rho} = (2^d+1) V^\dagger \ketbra{\hat{b}}V - I.
\end{align*}
As a result the sample complexity of CST is bounded as   $$n=O\qty(\frac{1}{\epsilon^2}\log m \max_j \tr{O^2_j}),$$
and the CST algorithm runs  $\tilde{O}\qty(2^{\Theta(d)}m \log m)$ classical time.

\section{Technical Lemmas and Proofs}\label{sec:technical}
In this section we present some of the technical lemmas and their proofs.

\subsection{Proof of Lemma \ref{lem:tr Hs}}\label{subsec:proof:lem:tr Hs}
\begin{customlemma}{\ref{lem:tr Hs}}
    The partial derivative observables $H_\bfs$ satisfy  $\tr{H_\bfs^2}\leq 4\tr{O^2}$ for any $\bfs\in \set{0,1,2,3}^n$. 
\end{customlemma}
\begin{proof}
    For the proof of the second claim note that
    \begin{align*}
        H_\bfs^2 = R_{\bfs,+}^\dagger O^2 R_{\bfs,+} + R_{\bfs,-}^\dagger O^2 R_{\bfs,-} - R_{\bfs,+}^\dagger O R_{\bfs}(\pi) O R_{\bfs,-} - R_{\bfs,-}^\dagger O R_{\bfs}(-\pi) O R_{\bfs,+}
    \end{align*}
    Taking the trace of this operator gives 
    \begin{align*}
        \tr{H_\bfs^2} =2 \tr{O^2} - \tr{ O R_{\bfs}(\pi) O R_{\bfs}(-\pi)} - \tr{O R_{\bfs}(-\pi) O R_{\bfs}(\pi)}.
    \end{align*}
    Note that $R_{\bfs}(\pi)= - i\qps$ and $R_{\bfs}(-\pi)=  i\qps$. To see this, note that $\qps = \sum_b \lambda_b \ketbra{b}$, where $\ket{b}$ is an eigenstate of $\qps$ and $\lambda_b \in \{-1,1\}$ is the corresponding eigenvalue. As a result,
    \begin{align*}
        R_{\bfs}(\pi) = e^{-\frac{\pi}{2}\qps} = \sum_{b} e^{-i\frac{\pi}{2}\lambda_b} \ketbra{b} =  \sum_{b} (i)^{-\lambda_b} \ketbra{b} = -i \qps
    \end{align*}
    where we used the fact that $i^x = xi$ for $x=\pm 1$. Similarly, $R_{\bfs}(-\pi)=i\qps$. Therefore,   
    \begin{align*}
        \tr{ O R_{\bfs}(\pi) O R_{\bfs}(-\pi)} = \tr{ O R_{\bfs}(-\pi) O R_{\bfs}(\pi)} = \tr{ O \qps O \qps}. 
    \end{align*}
    The observable $O$ has a Pauli decomposition of the form $O=\sum_{\bft \in \set{0,1,2,3}^d} a_\bft \sigma^\bft$ with $a_\bft\in \RR$. Then the above quantity is equal to 
    \begin{align*}
        \tr{ O \qps O \qps} = \sum_\bft \sum_\bfr a_\bft a_\bfr \tr{\sigma^\bft \qps \sigma^\bfr \qps}.
    \end{align*}
    Note that the trace term in the summation is zero unless $\bft=\bfr$. Also note that for any $i,j \in \{0,1,2,3\}$
    \begin{align*}
        \tr{\sigma^i\sigma^j\sigma^i\sigma^j} = \begin{cases}
            -2 & \text{if}~ i\neq j\neq 0\\
            2 & \text{otherwise.}\\
        \end{cases}
    \end{align*}
    Therefore, the above summation equals to 
    \begin{align*}
        \tr{ O \qps O \qps} =  2^n \sum_\bft (a_\bft)^{2} (-1)^{\abs{\set{i: t_i\neq s_i\neq 0}}}.
    \end{align*}
   On the other hand from a Parseval-type identity we know that $\tr{O^2}=  2^n \sum_\bft (a_\bft)^{2}$.  Putting everything together we have the expression for $\tr{H_\bfs^2}$:
   \begin{align*}
    \tr{H_\bfs^2} &= 2^{n+1} \sum_\bft (a_\bft)^{2} \qty(1-(-1)^{\abs{\set{i: t_i\neq s_i\neq 0}}})\\
                & \leq 2^{n+2} \sum_\bft (a_\bft)^{2}\\
                & = 4\tr{O^2}.
   \end{align*}
   where the last inequality holds because $1-(-1)^x \leq 2$ for any non-negative integer $x$.
\end{proof}

\subsection{Proof of Lemma \ref{lem:Bnorm-subgroup}}\label{subsec:proof:lem:Bnorm-subgroup}
\begin{customlemma}{\ref{lem:Bnorm-subgroup}}
Let $V\in \CC^{p\times p}$ be the matrix from Theorem \ref{thm:subgroup}, and $f(z) := (1-e^{-z})/z$ for $z\neq 0$ and $f(0) := 1$. Then,
\[
\|f(V)\| \le \frac{e^{\|V\|}-1}{\|V\|}.
\]
Moreover, if $|a_{\bfs}|\le c$ for all $\bfs\in\CS$, then
\[
\|f(V)\| \le \frac{e^{2pc}-1}{2pc}.
\]
\end{customlemma}

\begin{proof}
The power series expansion of $f$ around $0$ is given by
\[
f(z)= -\sum_{k=0}^\infty \frac{z^k}{(k+1)!}.
\]
Hence, for every matrix $V\in \CC^{p\times p}$,
\[
f(V)= -\sum_{k=0}^\infty \frac{V^k}{(k+1)!}.
\]

On the other hand, from the power series expansion of the matrix exponential $e^{tV}$, we have
\[
\int_0^1 e^{tV}\,dt
=
\int_0^1 \sum_{k=0}^\infty \frac{t^kV^k}{k!}\,dt
=
\sum_{k=0}^\infty \frac{V^k}{k!}\int_0^1 t^k\,dt
=
\sum_{k=0}^\infty \frac{V^k}{(k+1)!}.
\]
Thus
\[
f(V)=-\int_0^1 e^{tV}\,dt.
\]

Taking operator norms and using the triangle inequality,
\[
\|f(V)\|
\le \int_0^1 \|e^{tV}\|\,dt.
\]
Moreover, by the power-series expansion and submultiplicativity of the operator norm,
\[
\|e^{tV}\|
\le
\sum_{k=0}^\infty \frac{\|tV\|^k}{k!}
=
e^{t\|V\|}.
\]
Hence
\[
\|f(V)\|
\le
\int_0^1 e^{t\|V\|}\,dt
=
\begin{cases}
\dfrac{e^{\|V\|}-1}{\|V\|}, & \|V\|>0,\\[1ex]
1, & \|V\|=0.
\end{cases}
\]

It remains to bound $\|V\|$. Fix a column index $j$. Since the map
$\bfs\mapsto \bfs\oplus \bfs_j$ is a permutation of $\CS$, each coordinate of
$\bfv_j$ has magnitude $2|a_{\bfs}|$ for exactly one $\bfs\in\CS$. Therefore,
\[\sum_{r=1}^p |V_{rj}| = 2\sum_{\bfs\in\CS}|a_{\bfs}|.\]
This implies that the induced $\ell_1$-norm of $V$ equals $\norm{V}_{\ell_1}= 2\sum_{\bfs\in\CS}|a_{\bfs}|$.

The same argument applied row-wise shows	
\[\norm{V}_{\ell_\infty}=\max_{j}\sum_{r=1}^p |V_{rj}|=\max_{r}\sum_{j=1}^p |V_{rj}|= 2\sum_{\bfs\in\CS}|a_{\bfs}|.\]
Thus, from the identity $\|V\| \le \sqrt{\|V\|_{\ell_1}\|V\|_{\ell_\infty}}$, we have
\[\|V\| \le 2\sum_{\bfs\in\CS}|a_{\bfs}|.\]
Therefore, assuming $|a_{\bfs}|\le c$ for all $\bfs\in\CS$, we have
$\|V\| \le 2pc.$ Substituting this into the previous estimate gives the desired bound.
\end{proof}

\end{document}